\documentclass[prl,nofootinbib,floats,aps,twocolumn,tightenlines,superscriptaddress]{revtex4-1}
\usepackage[utf8]{inputenc}
\usepackage{amsmath}
\usepackage{amsfonts}
\usepackage{dsfont} 
\usepackage{graphicx}
\usepackage{todonotes}
\usepackage{gensymb}
\usepackage{bbold}
\usepackage{bm}
\usepackage{natbib}
\usepackage{amssymb}
\usepackage{latexsym}
\usepackage{fancyhdr}
\usepackage[bbgreekl]{mathbbol}
\usepackage{epsfig}
\usepackage{color}
\usepackage{textcomp}
\usepackage{amsthm}
\usepackage{comment}
\usepackage{ulem}
\usepackage{hyperref}

\usepackage{mathrsfs}

\newtheorem{theorem}{Theorem}[section]
\newtheorem{corollary}{Corollary}[theorem]
\newtheorem{lemma}[theorem]{Lemma}
\newtheorem{proposition}[theorem]{Proposition}

\newcommand{\ket}[1]{\left| {#1} \right\rangle}

\newcommand{\ketbra}[2]{\left| {#1} \right\rangle \!\! \left\langle {#2} \right|}
\newcommand{\prjct}[1]{\mathinner{|{#1}\rangle}\!\!\mathinner{\langle{#1}|}}
\newcommand{\tr}[1]{\mbox{$\mathrm{Tr}\left(#1\right)$}}
\newcommand{\Tr}[1]{\mathrm{Tr}\!\left(#1\right)}

\newcommand{\Id}{\mathds{1}}

\newcommand{\ii}{\mathrm{i}}

\newcommand{\cB}{\mathcal{B}}

\newcommand{\cE}{\mathcal{E}}
\newcommand{\cF}{\mathcal{F}}
\newcommand{\cH}{\mathcal{H}}
\newcommand{\sA}{\mathscr{A}}
\newcommand{\sB}{\mathscr{B}}
\renewcommand{\t}[1]{\mathrm{#1}}
\newcommand{\be}{\begin{equation}}
\newcommand{\ee}{\end{equation}}
\newcommand{\tot}{\!\otimes\!}

\begin{document}

\title{Certifying the building blocks of quantum computers from Bell's theorem}
\author{Pavel Sekatski}
\thanks{These authors equally contributed}
\affiliation{Quantum Optics Theory Group, Universit\"at Basel, Klingelbergstraße 82, CH-4056 Basel, Switzerland}
\affiliation{Institut für Theoretische Physik, Universit\"at Innsbruck, Technikerstraße 21a, A-6020 Innsbruck, Austria}
\author{Jean-Daniel Bancal}
\thanks{These authors equally contributed}
\affiliation{Quantum Optics Theory Group, Universit\"at Basel, Klingelbergstraße 82, CH-4056 Basel, Switzerland}
\author{Sebastian Wagner}
\affiliation{Quantum Optics Theory Group, Universit\"at Basel, Klingelbergstraße 82, CH-4056 Basel, Switzerland}
\author{Nicolas Sangouard}
\affiliation{Quantum Optics Theory Group, Universit\"at Basel, Klingelbergstraße 82, CH-4056 Basel, Switzerland}

\date{\today}
\begin{abstract}
\noindent The power of quantum computers relies on the capability of their components to maintain faithfully and process accurately quantum information. Since this property eludes classical certification methods, fundamentally new protocols are required to guarantee that elementary components are suitable for quantum computation. These protocols must be device-independent, that is, they cannot rely on a particular physical description of the actual implementation if one is to qualify a block for all possible usages. Bell's theorem has been proposed to certify, in a device-independent and robust way, blocks either producing or measuring quantum states. In this manuscript, we provide the missing piece: a method based on Bell's theorem to certify coherent operations such as storage, processing and transfer of quantum information. This completes the set of tools needed to certify all building blocks of a quantum computer. Our method distinguishes itself by its robustness to experimental imperfections, and so can be readily used to certify that today's quantum devices are qualified for usage in future quantum computers.

\end{abstract}

\maketitle

Experimental research on quantum computing is progressing at an unprecedented rate~\cite{Ladd10}. Five-qubit quantum computations combining around a dozen of elementary operations -- called quantum logical gates -- can nowadays be performed with a mean gate fidelity of $\sim$98\% using trapped ions~\cite{Debnath16} or superconducting circuits~\cite{Barends14}. In the next years, an on-going race between savvy research groups worldwide is expected to lead to the first quantum computation going beyond the capabilities of the most powerful classical supercomputer~\cite{Castelvecchi17}. If one is to reach the point of designing a universal quantum computer~\cite{DiVincenzo00}, it is crucial to proceed in a scalable way and certify that each new component is qualified for use in a quantum computer, independently of the purpose for which that larger device is used. 

Such a certification must be device-independent, that is, it cannot rely on a physical description of the actual implementation. Indeed, an exhaustive model of the setup is challenging, if not impossible, to establish. Relying on any particular model therefore amounts to making assumptions about the functioning of blocks. But seemingly harmless assumptions, on the underlying Hilbert space dimension for instance, can completely corrupt conclusion related to specific applications such as quantum-based secured communication~\cite{Acin06,Lydersen10}. Blocks certified in such a manner thus cannot be used safely for arbitrary purposes.

Bell's theorem~\cite{Bell64} has lead to device-independent certification schemes for components either producing quantum states or performing quantum measurements~\cite{Popescu92,Braunstein92,Mayers04,Tomamichel13,Bardyn09,Miller13,Yang14,Kaniewski16,Bamps16,Natarajan17,Coladangelo17,Supic17}. But these are just some of the elementary blocks needed to build a quantum computer (see Figure~\ref{fig:architecture}). In particular, a device-independent method assessing the quality of components in charge of the transfer, processing and storage of quantum information is still missing. Together with existing techniques, such a method would in principle allow for the certification of all kinds of elementary building blocks needed in a quantum computer.

\begin{figure}
\includegraphics[width=0.4\textwidth]{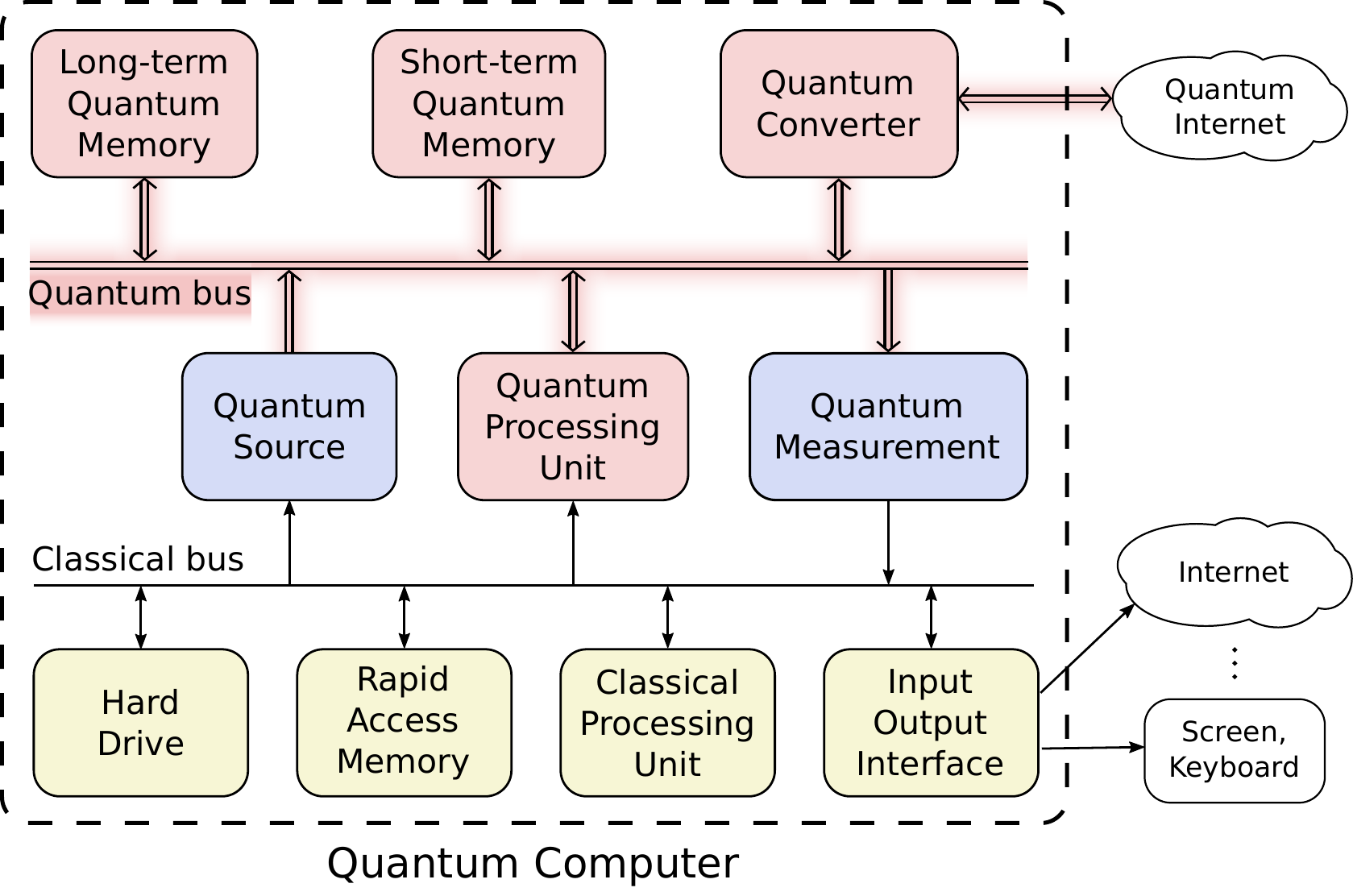}
\caption{Possible architecture of a future universal quantum computer (see also~\cite{Brennen03,Mariantoni11}). Elements in yellow are classical, and thus well characterized. Blue elements already admit device-independent certification schemes. Here, we demonstrate how to certify the components in red. In practice, several blocks may be merged into a single physical unit.}
\label{fig:architecture}
\end{figure}

Here, we fill this gap by showing how to certify a trace preserving quantum channel acting on one or several systems, that is, a general transformation taking quantum states and returning other quantum states. Our approach involves no description of the internal functioning of either the tested channel or the certification setup, but relies on the device-independent characterization of two entangled states, the first one serving as input to the channel, the second one being the output state. Interestingly, we can use state certifications that are robust to experimental imperfections to certify channels robustly.

Our goal is in sharp contrast with a line of research aiming to the certification of quantum computations~\cite{McKague13,Reichard13,Hajdusek15,Fitzsimons15,Coladangelo17b}. Our work addresses elementary blocks of a quantum computer and certify that they are qualified for use in future larger quantum devices. It builds on the work of Magniez et al.~\cite{Magniez06}, but differs in its formulation, methodology and robustness, see Appendix E for details. In particular, the robustness of our results is compatible with current technological capabilities, as we show below. In opposition to Ref.~\cite{Dall'Arno17}, we provide lower bounds on the quality of the blocks. Detailed recipes are given to certify the unitarity of one-qubit channels as well as two-qubit entangling operations. These recipes can be readily used in present-day experiments to certify transmission lines between processing and storage areas, storage devices, converters between various information carriers and arbitrary two-qubit controlled-unitary gates independently of the details and imperfections of the actual implementation.\\ 

\noindent \textit{Device-independent certification of a quantum channel--} 
We start by providing a definition of the device-independent certification of quantum channels. For this, we consider a 
scenario with two sides $\sA$ and $\sB$, each side containing potentially several parties depending on the channel to be certified. Each party performs measurements on one part of a shared state $\rho\in L(\cH_\sA\otimes\cH_\sB)$ and records the result of each experimental run. In addition, the parties on side $\sA$ have the freedom to decide whether or not to apply the channel to be certified $\cE$, an endomorphism on states in $\cH_\sA$, before performing the measurements (see Figure~\ref{fig:bigfigure} $a$ and $c$). The sources preparing the initial state, the measurement devices and the channel are treated as black boxes and the parties do not communicate with each other. The partial state prepared by the source at side $\sA$ is denoted $\rho_\sA =\text{Tr}_\sB \,\rho$.  

We say that the channel $\cE$ is certified device-independently if  the sole knowledge of the results given the measurement choices implies the existence of local isometries $\Phi_i:\cH_\sA\otimes\cH_i\to \cH_\sA \otimes \cH_i^\t{ext}$ and $\Phi_o: \cH_\sA \to \cH_o\otimes \cH_o^\t{ext}$  such that
\begin{equation}
\nonumber
(\Phi_o \circ \cE \circ \Phi_i \otimes \Id)\left[\rho_\sA\otimes\ketbra{\phi^+}{\phi^+}\right] = (\overline \cE \otimes \Id)\left[\ketbra{\phi^+}{\phi^+}\right] \otimes \rho_\text{ext}^{(i,o)},
\end{equation}
where $\overline \cE$ is the reference  channel mapping states from Hilbert space $\cH_i$ to the Hilbert space $\cH_o.$ Here, $\ket{\phi^+}$ is a maximally entangled state in $\cH_i\otimes \cH_i$, and $\rho_\text{ext}^{(i,o)}$ is some irrelevant residual state on $\cH_i^\t{ext}\otimes \cH_o^\t{ext}$. We emphasize that in device-independent certification, assumptions are made neither on the system’s state on which $\cE$ operates, nor on the dimension of the underlying Hilbert space. The local isometries $\Phi_i/\Phi_o$ identify subspaces/subsystems in which the channel $\cE$ acts exactly as the reference channel $\overline{\cE}$.

When the above equality does not hold exactly we quantify the relation between the channels $\cE$ and $\overline{\cE}$ through the following fidelity
\begin{equation}\label{eq:FE}
\cF(\cE,\overline\cE) = \underset{\Lambda_i,\Lambda_o}{\max}\ F\Big((\Lambda_o\!\circ\cE\circ\Lambda_i\otimes \Id)\left[\rho_\sA\otimes \prjct{\phi^+} \right], \overline\rho\Big).
\end{equation}
Here, $F(\rho,\sigma)=\tr{\sqrt{\sigma^{1/2}\rho\,\sigma^{1/2}}}$ is the Uhlmann fidelity. $\Id$ acts on the second half of $\ket{\phi^+}.$ $\Lambda_{i}[\cdot]= \t{Tr}_{\cH_i^\t{ext}}(\Phi_{i}[\cdot])$ traces out all degrees of freedom which are not in the preimage of $\cE$ while $\Lambda_o[\cdot]= \t{Tr}_{\cH_o^\t{ext}}(\Phi_o[\cdot])$ traces out all degrees of freedom which are not in the image of $\overline \cE.$ $\overline\rho = (\overline \cE \otimes \Id)\left[\ketbra{\phi^+}{\phi^+}\right].$ (See Fig.~\ref{fig:channelFidelity} and Appendix~A.1 and A.2 for details.)

This fidelity, which is optimized over all maps, can be understood as an extension of the Choi fidelity to device-independent scenarios. It guarantees  that the channel $\cE$ can be used to play the role of $\overline\cE$ in any circumstance with fidelity $\cF$. The maps achieving this fidelity describe the recipe for how to do that. Furthermore, the fidelity $\cF$ of Eq. \eqref{eq:FE} can be used to bound the distance between the two channels through the diamond norm, which informs us on the highest probability to distinguish the two channels in a single shot upon acting on arbitrary states~\cite{BenAroya10}, see Appendix~A.3.

\begin{figure}
\includegraphics[width=0.4\textwidth]{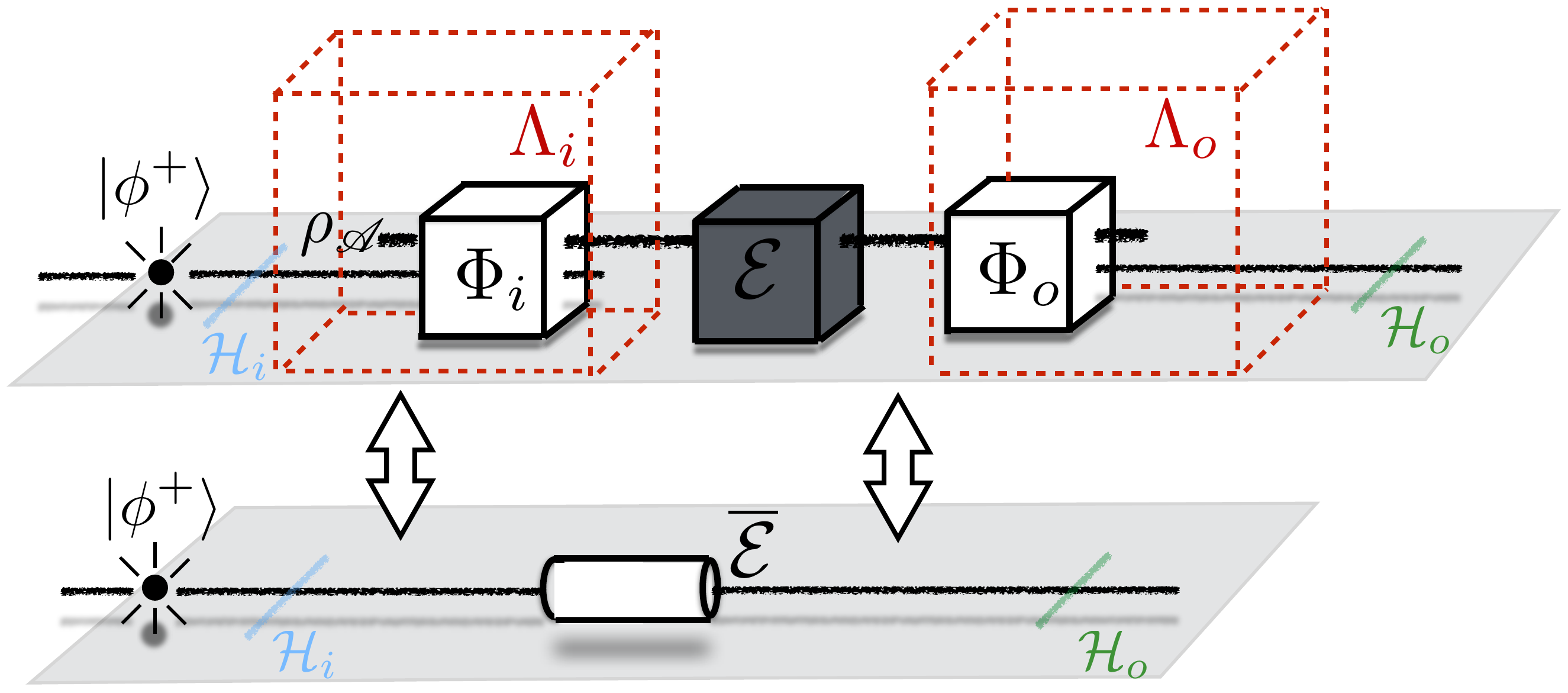}
\caption{Comparison between an unknown channel $\cE$ and a reference channel $\overline\cE$ operating on a Hilbert space $\cH_i.$ Half a maximally entangled state belonging to $\cH_i \otimes \cH_i$ is presented to $\cE$ by a local map $\Lambda_i$ that can also act on the initial quantum state $\rho_\sA$. Degrees of freedom which are not transmitted to the channel at this point are discarded. A local map $\Lambda_o$ is then used at the output of the channel $\cE$ to remove extra systems and extract the state of a subsystem to be compared with the Choi state $\overline\rho = (\overline \cE \otimes \Id)\left[\ketbra{\phi^+}{\phi^+}\right]$ of the reference channel. The channel fidelity $\cF(\cE,\overline\cE)$ is then obtained by maximizing the overlap between $\overline\rho$ and the channel output over all possible input isometries and output maps.}
\label{fig:channelFidelity}
\end{figure}

In the case where the target channel $\overline\cE$ acts on several parties, we distinguish these parties $\{A^{(k)}\}$ on the side $\sA.$ The input and output Hilbert spaces then have a tensor structure $\cH_{i/o}=\bigotimes_{k=1}^n \cH_{i/o}^{(k)}$ and the same is required from the maps $\Lambda_i$ and $\Lambda_o$, as spelled out in Appendix~A.2. \\

\noindent\textit{A practical device-independent bound on the channel fidelity --} Now that the goal is well established, we show that a channel certificate can be obtained by combining two certifications, one for the state serving as input of the channel and one for the output state, that is 
\begin{eqnarray}
F^i&=&F((\widetilde{\Lambda}^\sA_i\!\otimes\!\Lambda^\sB)[\rho],\ketbra{\phi^+}{\phi^+})\label{eq:Fi}\\
F^o&=&F((\Lambda^\sA_o\!\otimes\!\Lambda^\sB)[(\cE\!\otimes\!\Id)[\rho]],(\overline\cE\!\otimes\!\Id)[\ketbra{\phi^+}{\phi^+}]).
\label{eq:Ff}
\end{eqnarray}
$F^i$ corresponds to the fidelity of the input state $\rho$ with respect to the maximally entangled state $\ket{\phi^+}$. $F^o$ is the fidelity of the output state with respect to the image of $\ket{\phi^+}$ under the reference channel. As before, the role of the maps $\widetilde{\Lambda}^\sA_i$, $\Lambda^\sA_o$, and $\Lambda^\sB$ is to identify subspaces where the system states and the reference states can be compared, and the underlying isometries are enforced to have a product structure with respect to the partition of $\sA$ into separate parties.

In Appendix~B, we use the triangle and processing inequalities for the fidelity as well as properties of the isometries in Eqs.~\eqref{eq:Fi}-\eqref{eq:Ff} to show that the device independent Choi fidelity given in Eq.~\eqref{eq:FE} can be bounded by
\begin{equation}\label{eq:triangle}
\cF(\cE,\overline\cE) \geq \cos\left(\arccos\left(F^i\right)+\arccos\left(F^o\right)\right).
\end{equation}
Importantly, the bound holds for channels acting on several parties, in which case the states in Eq.~\eqref{eq:Fi}-\eqref{eq:Ff} are multi-partite and the maps $\widetilde{\Lambda}^\sA_i$ and $\Lambda^\sA_o$ are products of local maps for each party.

Formula~\eqref{eq:triangle} provides a first rate result: It shows how two channels can be compared even though they operate on  Hilbert spaces with (possibly unknown) different dimensions. This relation is made possible by the fact that the map $\Lambda^\sB$ is identical in both equations Eq.~\eqref{eq:Fi} and \eqref{eq:Ff}. One way to guarantee that the map is the same is to obtain certificates for both states with the same measurement boxes on side $\sB$. If this is fulfilled, a robust bound on the channel fidelity is obtained as soon as the input and output states are certified robustly. Interestingly, there are several known results and methods for state certification that are robust to noise~\cite{Bardyn09,Miller13,Yang14,Kaniewski16,Bamps16,Supic17,Natarajan17}.
 
We now show how Ineq.~\eqref{eq:triangle} can be used for the robust certification of (i) a one-qubit unitary, (ii) a two-qubit quantum logical gates.\\

\noindent\textit{Device-independent certification of a single-qubit unitary channel} Ensuring that quantum information can be preserved for a certain time, transmitted to a remote location or faithfully mapped between different physical systems are fundamental requirements for computing. This encapsulates quantum memories as hard drives, RAM units or parts of a quantum processor, quantum transmission lines between different units of a computer, and quantum converters between different information carriers. All these elements are mappings between input and output qubits, either separated in time or space or carried by different physical systems, that are ideally modeled by the identity channel.

Applying the formalism presented earlier to $\overline{\cE}=\Id$ in dimension two, involves ideally a maximally entangled two-qubit state as input state (see Fig.~\ref{fig:bigfigure}$a$). As the reference channel does not alter the input, we assess the fidelities of both input and output with the Clauser, Horne, Shimony, and Holt (CHSH) test~\cite{CHSH69}. The condition that $\Lambda^{\sB}$ is identical in both situation is then naturally satisfied. Given the CHSH values $\beta^{i/o}$, it is possible to bound the state fidelity as \cite{Kaniewski16}
\begin{align}\label{eq:FCHSH}
    F^{i/o} \geq F_\t{CHSH}=\sqrt{\frac12 \left(1+\frac{\beta^{i/o}-\beta^*}{2\sqrt{2}-\beta^*}\right)} \, ,
\end{align}
where $\beta^*=\frac{2(8 + 7\sqrt{2})}{17} \approx 2.11$. Inserting these fidelities into Eq.~\eqref{eq:triangle}, yields a robust device-independent certification of one-qubit unitaries depicted in Fig.~\ref{fig:bigfigure}$b$. Examples confirming the robustness can be found in Appendix~C.

Remarkably, testing the input state is not necessary for the certification of a unitary channel. Indeed one can see the channel itself as part of the local isometry. Hence, it is always possible to define $\widetilde{\Lambda}^\sA_i$ such that the fidelity of the input state is at least as large as the output fidelity, i.e. $F^i \geq F^o$. This relation together with Eq.~\eqref{eq:triangle} give a bound on the channel fidelity $\cF \geq 2(F^o)^2-1$ in terms of the output fidelity alone.\\

\begin{figure*}
\includegraphics[width=0.8\textwidth]{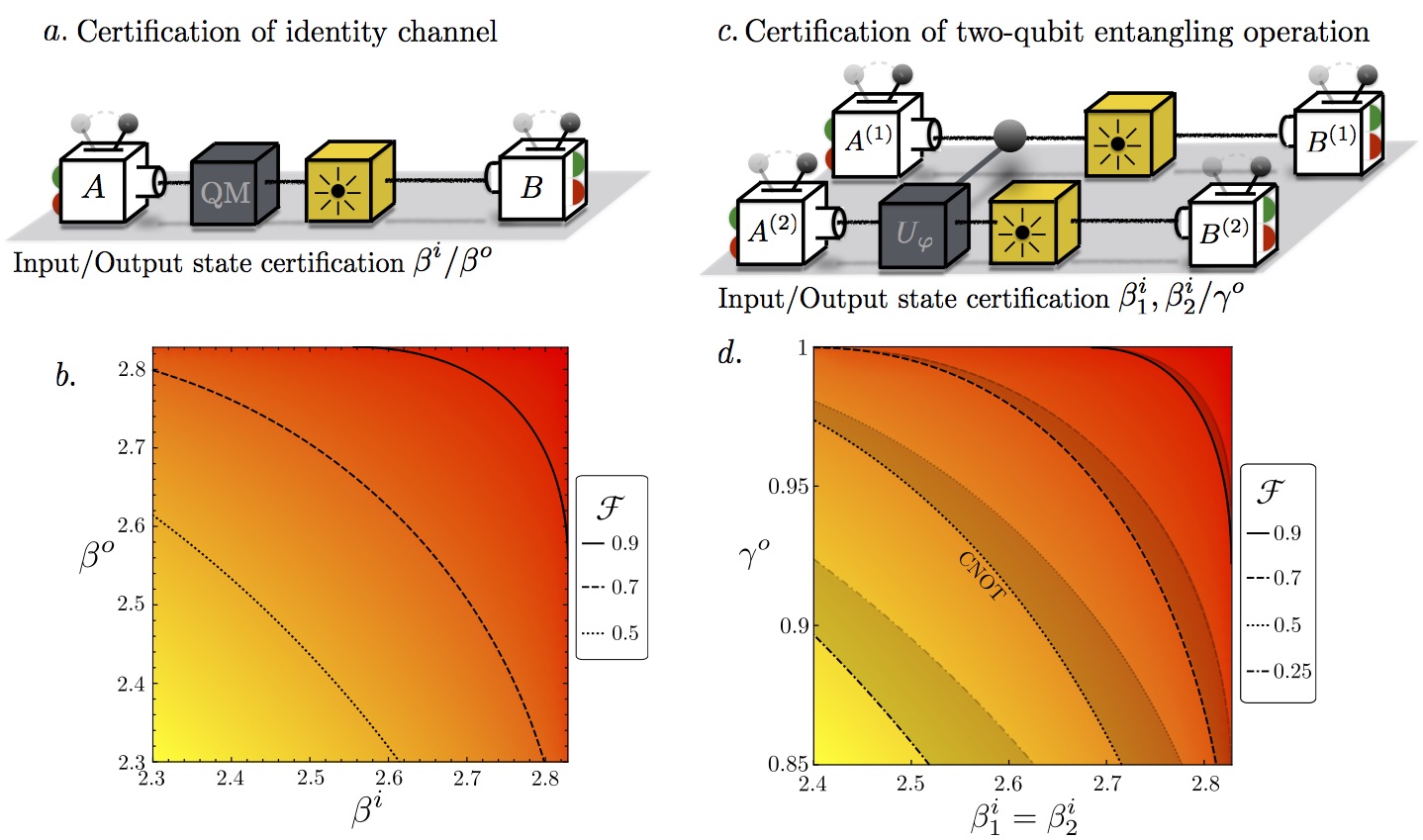}
\caption{Certification of one-qubit identity channel ($a. $ and $b.$) and two-qubit entangling operation ($c. $ and $d.$). $a.$ The certification of the identity in dimension 2 uses a source (yellow box) producing ideally a maximally two-qubit entangled state. The measurement devices (white boxes) $A$ and $B$ are used to perform two CHSH tests, with and without the tested device (black box). The sole knowledge of two CHSH values $\beta^i$ and $\beta^o$ gives a bound on the fidelity $\cF$ of the tested device with respect to the identity. $b.$ Robustness of the qubit identity certification as a function of the two CHSH values (color is a guide for the eye). $c.$ The certification of a two-qubit entangling operation uses a source (here represented with two yellow boxes) ideally producing two maximally entangled two-qubit states. Four measurement devices are used to perform Bell tests with and without the gate to be certified. The two Bell values $\beta^i_1$ and $\beta^i_2$ obtained to certify the two states produced by the source and the one obtained at the output of the gate (black box) $\gamma^o$ are used to bound the fidelity of any two-qubit controlled-unitary gates. $d.$ Robustness of the certification of two-qubit controlled-unitary operations (color is a guide for the eye). The best robustness is obtained for a class of gates including the Controlled-Not gate (CNOT). The grey lines shows the worst case. The greenish area thus includes all two-qubit controlled-unitary gates.}
\label{fig:bigfigure}
\end{figure*}

\noindent \textit{Device-independent certification of two-qubit entangling channels--}
Entangling gates are of central importance for quantum computing. On the one hand they are necessary for any non-trivial manipulation. On the other hand they are sufficient to enable universal quantum computation~\cite{Nielsen00}. We now present a setup that allows for the certification of an arbitrary two-qubit controlled-unitary gate. Such a gate can be put in the form
\be
CU_\varphi = \prjct{0}\otimes \Id + \prjct{1}\otimes e^{-\ii\, \varphi X}.
\ee
$CU_{\pi}$ is the controlled-NOT gate while $CU_{0}$ is the two-qubit identity channel. 

In order to bound the fidelity of an actual gate with the bipartite $CU_\varphi$ gate, we need to split side $\sA$ into two parties $A^{(1)}$ and $A^{(2)}$. Similarly, we also split side $\sB$ into $B^{(1)}$ and $B^{(2)}$ so that sharing a maximally entangled state of dimension 4 between $\sA$ and $\sB$ amounts to share two-qubit maximally entangled states $\ket{\phi^+_2}$ between $A^{(1)}$ and $B^{(1)}$ and between $A^{(2)}$ and $B^{(2)}$, c.f. Fig. \ref{fig:bigfigure}$c$. As we show now, four-partite statistics obtained after parties $A^{(1)}$ and $A^{(2)}$ jointly decide to use the device which supposedly performs the $CU_\varphi$ gate on their systems or not can lead to the certification of this gate. 

The certification of input states is identical to the one presented in the previous section. Performing two CHSH tests in parallel allows one to bound the fidelity of both states $F^{i}_{1(2)}\geq F_\t{CHSH}$ using Eq.~\eqref{eq:FCHSH} with the respective CHSH values $\beta^i_{1(2)}$. A lower bound on the fidelity of the global four party state is given by the product of the two singlet fidelities. Note that the maximal value for the CHSH test can be obtained with any pair of maximally non-commuting measurements on side $\sB.$ Hence, any certification of the output state with locally complementary observables for $B^{(1)}$ and $B^{(2)}$ automatically satisfies the constraint on the local map $\Lambda^\sB$.

The robust certification of the output state $\ket{\xi_\varphi} =(CU_\varphi\otimes\Id_{\sB}) \ket{\phi^+_2}^{\otimes 2}$ is less straightforward. To this end we devise a new family of Bell tests tailored to the robust certification of all states $\ket{\xi_\varphi}$, c.f. Appendix~D.4. Importantly, the Bell tests we derive have only two inputs and two outputs per party, and give the maximum quantum value for complementary measurement settings.

The fact that the Bell tests have  two inputs and two outputs per party allows us to make use of Jordan's Lemma in order to quantify their self-testing property. The latter ensures that the operator corresponding to the Bell test can be written as a direct sum of four qubit operators. Hence, we look for bounds on the fidelity assuming that the Bell operator is a four qubit operator, that is, qubit measurements are performed locally. If the extraction isometries only depend on local measurement settings and the square of the obtained fidelity bounds are convex functions of the mean value of the Bell operator, they automatically hold independently of the dimension~\cite{Scarani12}, see Appendix~D.2.
We find such bounds by using the isometries proposed in~\cite{Kaniewski16} which are known to provide very robust results for the singlet state. To do so we look for the state and measurement settings that minimize the fidelity of the extracted four qubit state with respect to $\ket{\xi_\varphi}$ while keeping a fixed expectation value $\gamma^o$ of the Bell operator, c.f. Appendix~D.3. The resulting bound on the fidelity is given by
\begin{align}
    F^o \geq \sqrt{\frac12 \left(1+\frac{\gamma^o-\gamma^*}{1-\gamma^*}\right)},
\end{align}
where $\gamma^*$ is a constant depending on the gate to be tested. This constant is upper bounded by $0.85$ for all $\varphi$, c.f. Appendix~D.5. 
Note that our approach to find Bell inequalities and deduce the corresponding robust fidelity bounds is applicable to other N-qubit states.

Given the bounds on the fidelities $F_1^i, F_2^i$ of the two initial states and on the fidelity $F^o$ of the output state, and checking that they have been obtained with common measurements for parties $B^{(1)}$ and $B^{(2)},$ we get from Ineq.~\eqref{eq:triangle} the following bound on the fidelity between the actual gate $\cE$ and the reference gate $\overline\cE = CU_\varphi$ 
$$
\cF(\cE,\overline\cE)\geq \cos(\arccos(F_1^{i} F_2^{i}) + \arccos(F^o)).
$$  
The result is shown in Fig.~\ref{fig:bigfigure}$d$ as a function of the observed Bell values assuming $\beta^i_1=\beta^i_2.$ Examples illustrating the robustness can be found in Appendix~C.

In analogy with the one-qubit identity certification, it is possible to prove that the actual two-qubit gate acts as a global unitary on side $\sA$ from $F^o$ only using $\cF \geq 2(F^o)^2-1$. This information alone is however not sufficient to identify the gates $CU_{\varphi}$ up to local isometries without additional assumptions, because the final state $\ket{\xi_\varphi}$ could be directly prepared by the source and merely transmitted by the device to be certified. \\

\noindent\textit{Discussions--}  We have introduced a framework for the device-independent certification of quantum channels. We applied our methods to two families of channels, namely single qubit identity channels and two qubit controlled unitary operations. They are of key importance for quantum computing and quantum networks and the robustness of our recipes insures that they can readily be used in present-day experiments.\\

\begin{acknowledgments}
This work was supported by the Swiss National Science Foundation (SNSF), through the NCCR QSIT, Grant PP00P2-150579, P300P2-167749 and 200021-175527. We also acknowledge the Army Research Laboratory Center for Distributed Quantum Information via the project SciNet.\\
\end{acknowledgments}

\noindent\textbf{Appendix A} \textit{Device-independent certification of quantum channels--}
In this appendix we give a thorough description of the device-independent channel certification introduced in the main text. There are three sections. Section A.1 provides a detailed definition of the device-independent certification of quantum channels and comments on Eq. (1) of the main text. Section A.2 addresses the extension to multi-party scenarios.  Section A.3 shows how the Choi fidelity bounds the diamond norm between the two channels.\\

\textbf{Appendix A.1} \textit{Formal definition of device-independent channel certification--}
In full generality, a quantum channel $\cE$ is a completely positive trace preserving map between linear operators on two Hilbert spaces \be
\cE: L(\cH)\to L(\cH').
\ee
It maps a quantum state $\rho$ to
\begin{equation}
\cE[\rho] = \sum_{i=1}^{I} K_i \rho K_i^\dag,
\end{equation}
where the Kraus operators $\{K_i\}_{i=1\ldots I}$ are represented by $\t{dim}(\cH')\times \t{dim}(\cH)$  complex matrices satisfying the relation $\sum_i K_i^\dag K_i = \Id$. In the following we will assume that the input and output Hilbert spaces are the same $\cH= \cH'=\cH_\sA.$ However, all the results can be straightforwardly generalized to the case where they do not match. In a bi-partite Bell-type scenario with measurement observables $M_{a|x}$ and $M_{b|x}$ corresponding to input $x$ on side $\sA$ and $y$ on side $\sB$ and outcome $a$ and $b$ respectively, we say that a behavior $P$ certifies device-independently a reference channel $\overline \cE : L(\cH_i)\to L(\cH_o)$ if, for every quantum realization $(\rho,\{M_{a|x},M_{b|y}\},\cE)$ compatible with $P$, there exist two local isometries 
\begin{align}
    &\Phi_i : \cH_\sA\otimes \cH_i \to \cH_\sA \otimes \cH_i^\t{ext},\\ 
    &\Phi_o : \cH_\sA  \to \cH_o \otimes \cH_o^\t{ext}
\end{align}
such that
\begin{align}
\nonumber
\t{Tr}_\t{ext}\Big((\Phi_o\! \circ \cE \circ\! \Phi_i \otimes \Id)\left[\rho_\sA\otimes \ketbra{\phi^+}{\phi^+}\right]\Big) \nonumber \\= (\overline \cE \otimes \Id)\left[\ketbra{\phi^+}{\phi^+}\right],
\end{align}
where the trace of over all the external subsystems, i.e. on $\cH_i^\t{ext}\otimes \cH_o^\t{ext}$. Here, $\ket{\phi^+}$ is the maximally entangled state in $\cH_i\otimes \cH_i.$ The maps are applied on the first Hilbert space and the identities on the second one. $\rho_\sA= \t{Tr}_\sB \, \rho$ is the unknown state prepared by the source on side $\sA$. The role of the isometries is to identify a subspace for the channel input and a subsystem for the channel output between which the channel $\cE$ acts as desired. To put it differently, the isometries define a recipe
\begin{align}
\cE_{i\text{-}o}: &L(\cH_i)\to L(\cH_o)\nonumber\\&\varrho \mapsto \t{Tr}_\t{ext}\Big(\big(\Phi_o\! \circ \cE \circ\! \Phi_i\big)[\rho_\sA\otimes \varrho]\Big)
\end{align}
of how the channel $\cE$ can be used in order to perform the desired operation, i.e. $\cE_{i\text{-}o} = \overline\cE$, as depicted in Fig.~\ref{fig:singleinput}.

When the equality presented before does not hold exactly, we would naturally define the distance between the actual channel and the reference one through 
\begin{equation}
\label{correct_def}
\cF(\cE,\overline\cE) = \max_{\Phi_o,\Phi_i} F\Big( \big(\cE_{i\text{-}o}\otimes \Id\big) \left[\ketbra{\phi^+}{\phi^+}\right], \overline{\rho} \Big)
\end{equation}
where $F(\rho,\sigma)=\tr{\sqrt{\sigma^{1/2}\rho\,\sigma^{1/2}}}$ is the Uhlmann fidelity and $\overline \rho = (\overline \cE  \otimes \Id) [\prjct{\phi^+}]$
is the target state. Note that the fidelity is symmetric. 

\begin{figure}[t!]
\centering
\includegraphics[width=\columnwidth]{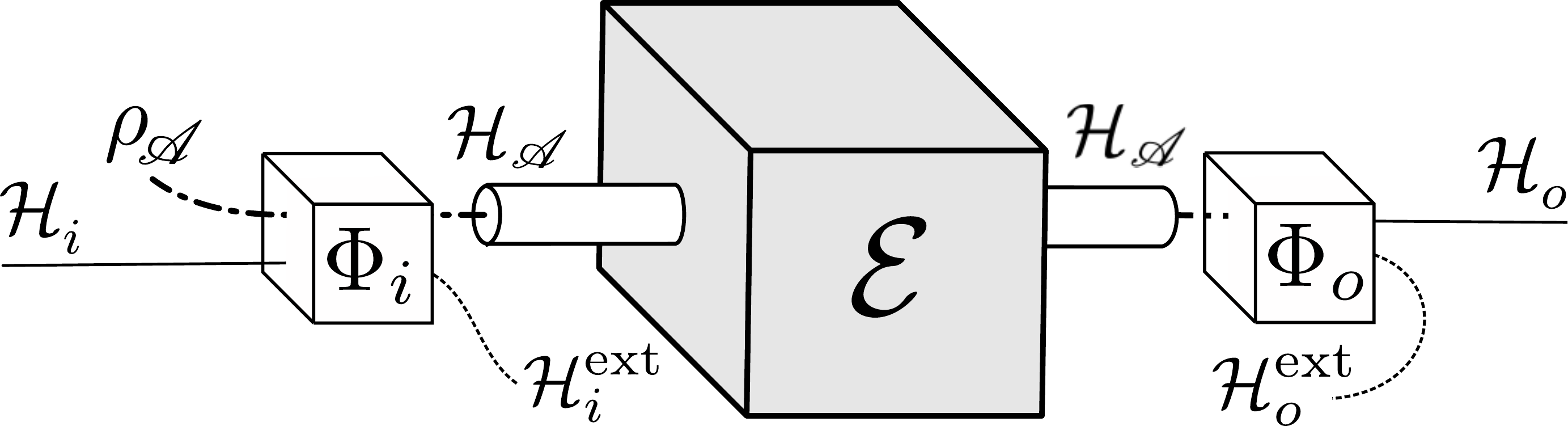}
\caption{The isometries $\Phi_i$ and $\Phi_o$ give a recipe of how the channel $\cE$ can be used in order to perform the desired operation $\overline \cE$ between $\cH_i$ and $\cH_o$.}
\label{fig:singleinput}
\end{figure}

In Proposition \ref{co:lemma1} below, we prove that the fidelity~\eqref{correct_def} can be related to a Uhlmann fidelity computed on the full Hilbert space, including the external systems $\rho_\t{ext}\in L(\cH_i^\text{ext}\otimes \cH_o^\text{ext})$:
\begin{align}
\label{pure_def}
&\cF(\cE,\overline\cE) \geq \nonumber \\
&\max_{\Phi_o,\Phi_i, \rho_{\text{ext}}} F\Big((\Phi_o \circ \cE \circ \Phi_i) \otimes \Id\left[\rho_\sA\otimes\ketbra{\phi^+}{\phi^+}\right], \overline\rho \otimes \rho_\text{ext}\Big),
\end{align}
with equality when the target channel $\overline\cE$ is unitary (i.e. the target state $\overline\rho$ is pure). This last expression is closer to the original definition of quantum state certification~\cite{Mayers04}. When the target channel $\overline\cE$ is not unitary, however, inequality \eqref{pure_def} is not tight, and the expressions in~\eqref{correct_def} better captures the relation between channels $\cE$ and $\overline\cE$. To see this, consider the example where the reference channel $\overline \cE$ is a totally depolarizing single-qubit channel. The channel to be certified $\cE$ can implement this depolarizing channel by entangling the system qubit with an external qubit. While $\cE$ and $\overline \cE$ operate identically on the system, $\cE$ outputs a pure global state. Moreover, the latter is entangled and the maximal fidelity cannot be obtained by optimizing it over a product state  $\overline \rho\otimes \rho_\t{ext}.$

To shorten the notation we  define the injection map 
\begin{align}
\Lambda_i: & L(\cH_\sA \otimes \cH_i)\to L(\cH_\sA)\\
& \varrho \mapsto \t{Tr}_{\cH_i^\t{ext}}\Big(\Phi_i[\varrho] \Big)
\end{align}
and extraction map
\begin{align}
\Lambda_o: & L(\cH_\sA )\to L(\cH_o)\\
& \varrho \mapsto \t{Tr}_{\cH_o^\t{ext}}\Big(\Phi_o[\varrho] \Big)
\end{align}
which allow to simply write $\cE_{i\text{-}o}[\bullet] = \Lambda_o\circ\cE\circ\Lambda_i[\rho_\sA\otimes\bullet]$.

\begin{lemma}\label{lemma:1}
Given three quantum states $\rho \in L(\cH_\t{sys}\otimes \cH_\t{ext})$, $\overline\rho= \ketbra{\psi}{\psi}\in L(\cH_\t{sys})$ and $\sigma\in  L(\cH_\t{ext})$ the following relation holds:
\begin{equation}
F\big(\rho,\overline\rho\otimes \sigma\big) = F\big(\rho_\t{sys},\overline\rho\big)\, F\big(\varrho_\t{ext},\sigma\big),
\end{equation}
where $\rho_\t{sys}=\t{Tr}_\t{ext}(\rho)$ and $\varrho_{\textnormal{ext}} = \frac{\t{Tr}_\t{sys}\left(\rho \,\,\overline\rho\otimes\Id\right)}{\t{Tr}\left(\rho \,\,\overline\rho\otimes\Id\right)}$.
\end{lemma}

\begin{proof}
Let us first note that $\prjct{\psi}^\frac{1}{2}=\prjct{\psi}$ and expand the fidelity:
\begin{equation}
\begin{split}
F&\big(\rho, \overline \rho \otimes \sigma \big)\\
&=\tr{\big(\prjct{\psi}\otimes\sigma\big)^{\frac12}\rho\, \big(\prjct{\psi}\otimes\sigma\big)^{\frac12}}^{\frac12}\\
&=\t{Tr}\Big(\big(\Id\tot\sqrt{\sigma}\big) \big(\prjct{\psi}\otimes\!\Id\big) \rho\, \big(\prjct{\psi}\otimes\!\Id\big) \big(\Id\tot\sqrt{\sigma}\big)\Big)^{\frac12}\\
&=\t{Tr}\Bigg(\big(\Id\tot\sqrt{\sigma}\big) 
\big(\prjct{\psi}\otimes \t{Tr}_\t{sys}(\rho\, \prjct{\psi}\otimes\!\Id)\big) \big(\Id\tot\sqrt{\sigma}\big)\Bigg)^{\frac12}\\\
&=\t{Tr}_\t{sys}\big(\prjct{\psi}\big)^{\frac12}\, \t{Tr}_\t{ext}\Big(\sqrt{\sigma} 
\t{Tr}_\t{sys}(\rho\, \prjct{\psi}\otimes\!\Id) \sqrt{\sigma}\Big)^{\frac12}\\
&=\sqrt{\t{Tr}(\rho\, \prjct{\psi}\otimes\!\Id)}\, \t{Tr}_\t{ext}\Big(\sqrt{\sigma} \,
\underbrace{\frac{\t{Tr}_\t{sys}(\rho\, \prjct{\psi}\otimes\!\Id)}{\t{Tr}(\rho\, \prjct{\psi}\otimes\!\Id)}}_{\varrho_\t{ext}} \sqrt{\sigma}\Big)^{\frac12}\\
&=\sqrt{\t{Tr}_\t{sys}\Big
(\t{Tr}_\t{ext}(\rho)\, \prjct{\psi}\Big)}\, F(\varrho_\t{ext},\sigma).\\
&=\sqrt{\t{Tr}_\t{sys}\Big
(\rho_\t{sys}\, \prjct{\psi}\Big)}\, F(\varrho_\t{ext},\sigma).
\end{split}
\end{equation}
Finally, it is easy to see that the term with the square root equals to $F(\rho_\t{sys}, \overline \rho)$. To this end expand
\be\nonumber
\begin{split}
F&(\rho_\t{sys}, \overline \rho) = \t{Tr} \sqrt{\prjct{\psi} \rho_\t{sys} \prjct{\psi}}=\\ &\tr{\prjct{\psi}} \sqrt{\tr{\rho_\t{sys}\, \prjct{\psi}}}=
\sqrt{\tr{\rho_\t{sys}\, \prjct{\psi}}}
\end{split}
\ee
which completes the proof.
\end{proof}


\begin{proposition}\label{co:lemma1}
Given a target state $\overline \rho \in L(\cH_\t{sys})$ and any state $\Phi[\rho]\in L(\cH_\t{sys}\otimes \cH_\t{ext})$ with $\Lambda[\rho]=\t{Tr}_\t{ext}\left(\Phi[\rho]\right)\in L(\cH_\t{sys})$, the following relation holds
\be\nonumber
F(\Lambda[\rho],\overline \rho) \geq \max_{\rho_\t{ext}} F(\Phi[\rho],\overline \rho\otimes \rho_\t{ext}).
\ee
Moreover, when the target state $\overline \rho ={\overline \rho}^2$ is pure, the maximum is attained for the state $\rho_\t{ext}=\frac{\t{Tr}_\t{sys}(\Phi[\rho]\, \overline \rho \otimes\Id )}{\t{Tr}(\Phi[\rho]\, \overline \rho \otimes\Id )}$ and the inequality is saturated.
\end{proposition}
\begin{proof}
First note that the processing inequality implies $ F(\Phi[\rho],\overline \rho\otimes \rho_\t{ext})\le F(\Lambda[\rho],\overline \rho) $ -- tracing out subsystems can only increase the fidelity. This implies the inequality. Second, when $\overline\rho$ is pure, we can apply Lemma \ref{lemma:1} to the fidelity $F(\Phi[\rho], \overline \rho \otimes \rho_\t{ext})$, which gives 
\be
F(\Phi[\rho],\overline \rho\otimes \rho_\t{ext})= F(\Lambda[\rho],\overline \rho) F(\varrho_\t{ext}, \rho_\t{ext}),
\ee
with $\rho_\t{ext}=\frac{\t{Tr}_\t{sys}(\Phi[\rho]\, \overline \rho \otimes\Id )}{\t{Tr}(\Phi[\rho]\, \overline \rho \otimes\Id )}$. Hence, the equality $F(\Lambda[\rho],\overline \rho) = F(\Phi[\rho],\overline \rho\otimes \rho_\t{ext})$
is obtained for the choice $\rho_\t{ext}=\varrho_\t{ext}$.
\end{proof}


\textbf{Appendix A.2} \textit{Extension to multi-partite scenarios--}
In the framework of the certification of channel with several inputs, e.g. two-qubit gates, the channel $\cE$ supposedly implements an interaction between a number of physically distinct subsystems, which can be clearly identified at the input and the output of the tested device. In this case, the side $\sA$ is composed of $n$ parties\footnote{Note that more generally the number of input and output subsystems can differ, but all the formalism straightforwardly generalizes to this case.} $\{A^{(k)}\}_{k=1}^n$ carrying one subsystem each, but also the reference channel $\overline \cE$ comes with a product structure for the Hilbert spaces $\cH_{i/o}=\bigotimes_{k=1}^n \cH_{i/o}^{(k)}$.
\begin{figure}[t!]
\centering
\includegraphics[width=\columnwidth]{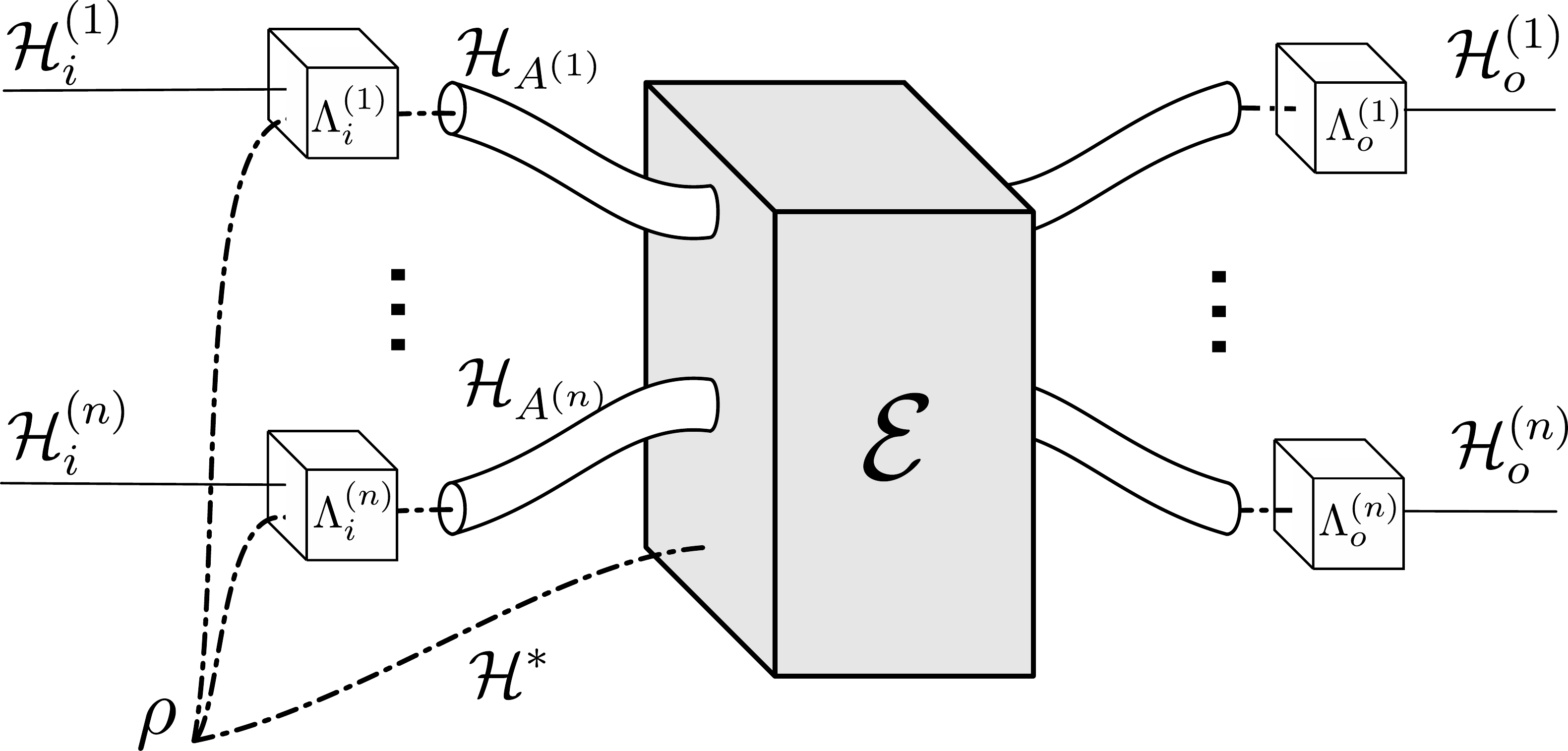}
\caption{Device independent certification of multipartite channels: decomposition of the input and output of the channel to be certified $\cE$ into Hilbert spaces belonging to the different parties $A^{(k)}$, and the role of injection/extraction maps. The induced map from $\cH_i$ to $\cH_o$ can be compared to the target channel $\overline\cE$.}
\label{fig:channel}
\end{figure}
In this context, a certification up to global isometries on side $\sA$ is not sufficient, e.g.  this would not distinguish an entangling two-qubit gate (a CNOT gate for example) with the identity. Hence, in order to allow for the device-independent certification of channels acting on several systems, the the maps $\Lambda_i$ and $\Lambda_o$ must be local with respect to the partition of $\sA$, as we now describe. 

A channel  $\cE :L(\cH_\sA) \to L(\cH_\sA)$ describing the tested gate operates on an Hilbert space $\cH_\sA$ of unknown dimension. Nevertheless, the gate acts on several subsystems that can be unambiguously separated. Hence, $\cH_\sA$ comes with a partition 
\be
\cH_\sA = (\bigotimes_{k=1}^{n} \cH_{A^{(k)}})\otimes \cH^*,
\ee
where each $\cH_{A^{(k)}}$ carries the state of the subsystem $k$ on which the device acts, and $\cH^*$ can carry external systems that define the initial state of the tested device. For example, the on/off button on the device can be such an external subsystem proper to the device. The Hilbert space for each physical subsystem $\cH_{A^{(k)}}$ has an unspecified dimension . We extend it with the input Hilbert space $\cH_i^{(k)}$ of the specified dimension, which will carry part of the maximally entangled state $\ket{\phi^+}\in\cH_i=\bigotimes_{k=1}^{n} \cH_i^{(k)}$. 

Then each local injection map sends state of this extended Hilbert space
\begin{align}
\Lambda_i^{(k)}: L\left(\cH_{A^{(k)}}\otimes\cH_i^{(k)}\right)\to L\left(\cH_{A^{(k)}}\right)
\end{align}
into the physically relevant Hilbert space $\cH_{A^{(k)}}$ on which the gate $\cE$ can act. With the overall map $\Lambda_i= (\bigotimes_{k=1}^n \Lambda_i^{(k)})\otimes \Id^*$, where $\Id^*$ simply says that the internal state of the device is untouched. Similarly, the local extraction given by
\be
\Lambda_o^{(k)}: L(\cH_{A^{(k)}}) \to L(\cH_o^{(k)})
\ee
maps the state of the physical output of the gate into the output Hilbert space of the ideal channel $\overline \cE$ such that the global map reads $\Lambda_o = \text{Tr}_{\cH^*} (\bigotimes_{k=1}^n \Lambda_o^{(k)})\otimes \Id^*$.

With these definitions, the channel composed with the maps takes the form
\be
\Lambda_o\circ \cE \circ \Lambda_i: L\left(\cH_\sA\otimes\cH_i \right)\to L(\cH_{o}).
\ee
By setting the state of the subsystem $\cH_\sA$ to be the one prepared by the source $\rho_\sA$, we finally get the desired recipe
\begin{align}
\cE_{i\text{-}o}: &L(\cH_i)\to L(\cH_o)\\
& \varrho \mapsto  \Lambda_o\circ \cE \circ \Lambda_i[\rho_\sA \otimes \varrho] \label{eq:recipe global}
\end{align}
that allows to compare $\cE$ to the target channel $\overline \cE$, as depicted in Fig.~\ref{fig:channel}.

Note that the state $\rho_\sA$ also contains the initial internal state of the device itself. This allows to address eccentric scenarios where the source can be entangled with the device. In practice, one can often assume that the state produced by the source is independent from the device, in which case the state $\rho_\sA$ decomposes as $\rho_\sA\otimes\rho^*$, where $\sA$ only refers to the source. Whenever this is possible, we can  forget about the internal state of the gate $\rho^*$ and the Hilbert space $\cH^*$, absorbing it in the definition of the channel $\cE$. We assume that this is the case for the rest of this section.

As a particular application of the above definition, let us now consider the case where the state in the experiment is produced by $n$ independent sources, each of which distributes an entangled state to $A^{(k)}$ and $B^{(k)}$. Under this assumption the marginal state of $\sA$ takes the form
\be
\rho_\sA = \bigotimes_{k=1}^n\rho_{A^{(k)}}.
\ee
Because each auxiliary state $\rho_{A^{(k)}}$ can be created locally, one can absorb them in newly defined injection maps
\begin{align}
    {\Lambda'}_i^{(k)}: &L(\cH_i^{(k)}) \to L(\cH_{A^{(k)}})\\
    & \varrho \mapsto  \Lambda_i^{(k)}[\rho_{A^{(k)}}\otimes\varrho].
\end{align}
With the maps defined this way, the explicit dependence on the state of the source $\rho_\sA$ disappears, that is
\be
\cE_{i\text{-}o} = \left(\bigotimes_{k=1}^n \Lambda_o^{(k)}\right) \circ \cE \circ\left(\bigotimes_{\ell=1}^n {\Lambda'}_i^{(\ell)}\right),\label{eq: recipe local}
\ee
as depicted in Fig.~\ref{fig:channelIndieSources}.

It is important to realize that if the sources are not independent such a formulation of the device-independent certification of a gate is impossible because one can easily imagine channels which can only perform the desired operation when the source provides an entangled state in some auxiliary degrees of freedom. For example, imagine that the source distributes two auxiliary maximally entangled states $\ket{\phi^+}_{A^{(k)}\text{-}A^{(\ell)}}$ to each pair of subsystems $A^{(k)}$ and $A^{(\ell)}$, and that the channel $\cE$ first teleports all the inputs on one subsystem, say $A^{(1)}$, performs the desired operation $\overline{\cE}$ locally at $A^{(1)}$, and then teleports the resulting states back to the respective parties. This sequence of operations would allow the channel to act as desired, but only provided that these singlets are actually distributed by the source. In other words, it would fail to work when provided independent auxiliary states. The recipe defined in Eq.~\eqref{eq: recipe local} would then fail to act as the target channel $\overline\cE$ and only the procedure defined in Eq.~\eqref{eq:recipe global}, which includes the partial state of the source $\rho_\sA$ produced in experiment, would perform the desired operation. \\

\begin{figure}[t!]
\centering
\includegraphics[width=\columnwidth]{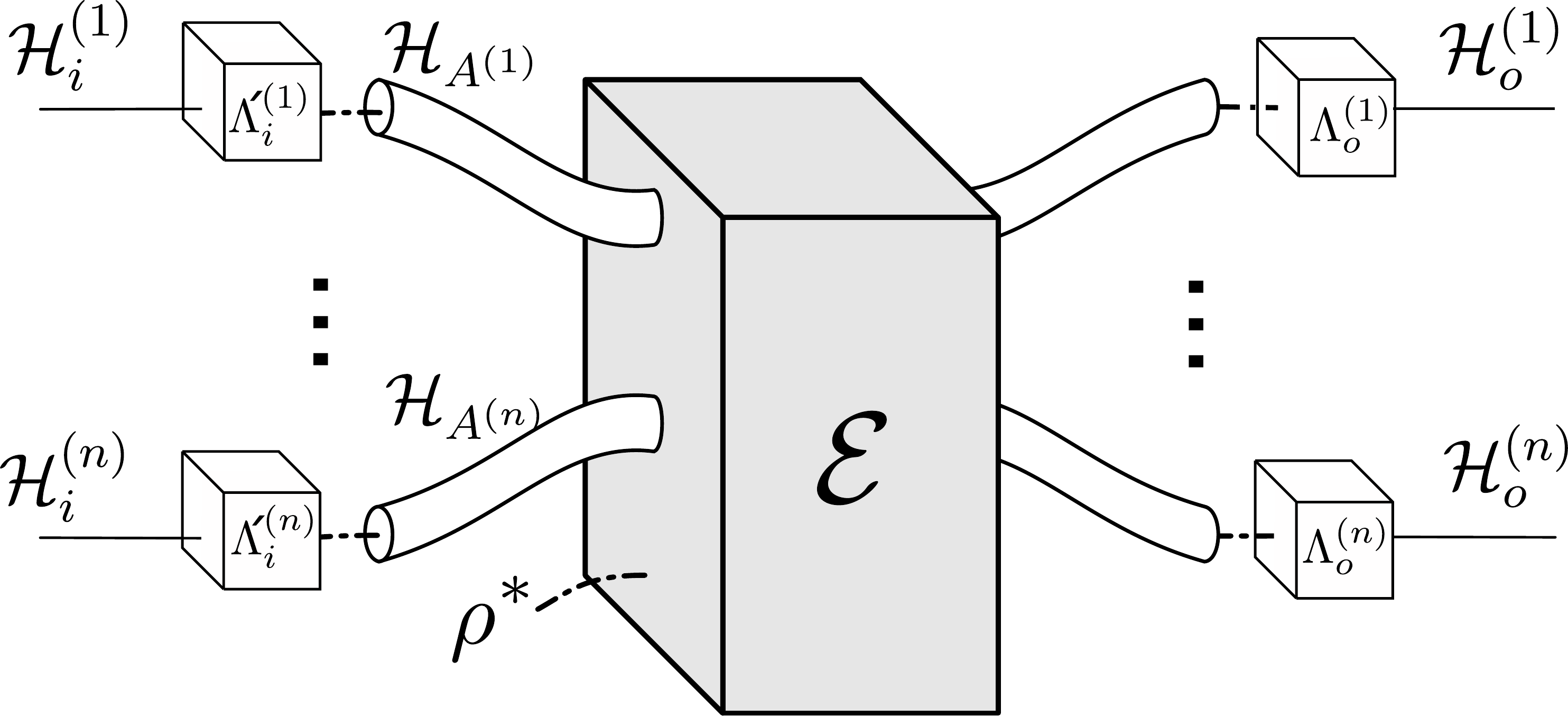}
\caption{Device independent certification of multipartite channels with $n$ independent sources. The local injection maps $\Lambda_i'^{(k)}$ can be optimized independently of the system's state $\rho$.}
\label{fig:channelIndieSources}
\end{figure}


\textbf{Appendix A.3} \textit{Relation between the Uhlmann channel fidelity and the diamond norm--}
We show how the Uhlmann fidelity between the Choi state of two channels can be used to bound the diamond norm between these two channels.

Let us first recall that the diamond norm between two channels $ \cE_{i\text{-}o}$ and $\overline \cE$, acting between the same Hilbert spaces $L(\cH_i) \to L(\cH_o)$, is defined as\footnote{Remark that one can sometimes see the definition with  the maximization running over states $\ket{\zeta}$ on an arbitrarily extended Hilbert space  $\cH_i\otimes \cH_\t{ext}$. But because of the Schmidt decomposition any such state only has support on a subspace of $\cH_\t{ext}$ of dimension $\t{dim}(\cH_i)$ at most. Hence it is sufficient to consider states $\ket{\zeta}\in \cH_i\otimes \cH_i$.}
\begin{align}
\nonumber
&||\cE_{i\text{-}o} - \overline \cE ||_\diamond =\\
&\sup_{\ket{\zeta}\in \cH_i\otimes \cH_i} || (\cE_{i\text{-}o} \otimes \Id)[\prjct{\zeta}]-(\overline\cE \otimes \Id)[\prjct{\zeta}]\,||_1.
\end{align}
This expression has an operational meaning: It directly relates to the maximal probability to discriminate the two channels in a single measurement, when comparing the images of some state $\ket{\zeta}$. 

We are interested to compare the reference channel $\overline \cE$ and the channel $\cE_{i\text{-}o}=\Lambda_o \circ \cE \circ \Phi_i$.
In particular, we show that the fidelity  $\cF(\cE,\overline \cE)$ of Eq.~(1) of the main text can be used to bound the diamond norm between $\cE_{i\text{-}o}$ and $\overline \cE$.

\begin{proposition}
The diamond norm $||\cE_{i\text{-}o} - \overline \cE ||_\diamond$ between two quantum channels 
\be
\cE_{i\text{-}o}\text{ and }\overline \cE : L(\cH_i)\to L(\cH_o)
\ee
is upper bounded by 
\be
||\cE_{i\text{-}o}- \overline \cE||_\diamond \leq  2\, \t{dim}(\cH_i)\sqrt{1-\cF^2(\cE_{i\text{-}o}, \overline \cE)},
\ee
where
\be\label{eq: F prop3}
\cF(\cE_{i\text{-}o}, \overline \cE) = F\big((\cE_{i\text{-}o}\otimes \Id)[\prjct{\phi^+}], (\overline \cE\otimes \Id)[\prjct{\phi^+}]\big)
\ee
is the Choi fidelity between the two channels defined with the maximally entangled state $\ket{\phi^+}\in \cH_i\otimes \cH_i$.
\end{proposition}


\begin{proof}
We start by  showing that any state $\ket{\zeta}\in \cH_i\otimes \cH_i$ can be probabilistically prepared from the maximally entangled state $\ket{\phi^+}$ by solely acting on the second Hilbert space. Given the Schmidt decomposition of the state $\ket{\zeta}= \sum_i \zeta_i \ket{i}\ket{i'}$ there exists a local unitary $\Id\otimes U$ such that
\be
\ket{\overline \phi^+}=\Id\otimes U \ket{\phi^+}= \frac{1}{\sqrt{\t{dim}(\cH_i)}}\sum_i \ket{i}\ket{i'}.
\ee
Furthermore, the operator $K= \sum_i \zeta_i \prjct{i'}$ applied on this state results into
\be
\Id\otimes K \ket{\overline \phi^+}= \frac{1}{\sqrt{\t{dim}(\cH_i)}} \ket{\zeta}.
\ee
$K$ is a valid Kraus operator $0\leq K^\dag K \leq 1$ and hence together with the unitary $U$ can be implemented as one branch of some probabilistic announced  two-branch protocol $\mathcal{P}$ acting on the second Hilbert space (a two-outcome POVM). From the initial state  $\ket{\phi^+}$ this protocol prepares the state $\ket{\zeta}$ with probability $p=\frac{1}{\t{dim}(\cH_i)}$ (successful branch), and some other state $\varrho$ with probability $(1-p)$ (failure branch).

Next recall that the Uhlmann fidelity between any two states $\rho$ and $\sigma$ upper bounds the trace-distance between the same states as
\be
D(\rho,\sigma)=\frac{1}{2}||\rho-\sigma||_1 \leq \sqrt{1-F^2(\rho,\sigma)}.
\ee
Hence, from Eq.~\eqref{eq: F prop3} one gets
\be
D\Big(( \cE_{i\text{-}o}\otimes \Id)[\prjct{\phi^+}],(\overline \cE\otimes \Id)[\prjct{\phi^+}]\Big)\leq \sqrt{1-\cF^2}.
\ee
The trace distance $D$ satisfies the processing inequality, meaning that it cannot increase whatever common operation is performed on the two states. In particular, this holds for the protocol $\mathcal{P}$, which, in addition, commutes with the channels and can therefore be applied before them, leading to
\begin{align}
D\Big(( \cE_{i\text{-}o}\otimes \Id)[\prjct{\phi^+}],(\overline \cE\otimes \Id)[\prjct{\phi^+}]\Big)\geq\\ 
p D\Big(( \cE_{i\text{-}o}\otimes \Id)[\prjct{\zeta}],(\overline \cE\otimes \Id)[\prjct{\zeta}]\Big)
+(1-p) D'.
\end{align}
$D'$ stands for the trace distance between the states prepared in the failure branch and satisfies $(1-p) D'\geq 0$. 
Combining with the previous inequality yields
\be
D\Big((\cE_{i\text{-}o}\otimes \Id)[\prjct{\zeta}],(\overline \cE\otimes \Id)[\prjct{\zeta}]\Big)\leq \frac{1}{p} \sqrt{1-\cF^2}.
\ee
Since this holds for any state $\ket{\zeta}$ and $p=\frac{1}{\t{dim}(\cH_i)},$ the proposition is proven.
\end{proof}

\vspace{20 pt}

\noindent\textbf{Appendix B} \textit{Bounding the channel fidelity with state fidelities--} In this Appendix, we prove Eq.~(4) of the main text. 

\begin{proposition}\label{proposition:B1}
Let $\rho$ be a quantum state shared between the sides $\sA$ and $\sB$, $\cE$ and $\overline \cE$ two channels acting on side $\sA$ with $\overline \cE : L(\cH_i)\to L(\cH_o)$ a reference channel, and $\ket{\phi^+}\in \cH_i\otimes\cH_i$ a maximally entangled state. Given the local maps $\widetilde{\Lambda}_i^\sA [\cdot]=\textnormal{Tr}_{\widetilde{\textnormal{ext}}} (\widetilde{\Phi}_i^\sA[\cdot])$,  $\Lambda_o^\sA[\cdot]=\textnormal{Tr}_\textnormal{ext}(\Phi_o^\sA[\cdot])$ and $\Lambda^\sB[\cdot]=\textnormal{Tr}_\textnormal{ext}(\Phi^\sB[\cdot])$ acting on $\sA$ and $\sB$ respectively with the corresponding isometries $\widetilde{\Phi}_i^\sA$, $\Phi_o^\sA$ and $\Phi^\sB$, the following two fidelities
\begin{eqnarray}
F^i&=&F\left((\widetilde{\Lambda}^\sA_i\otimes \Lambda^\sB) [\rho],\ketbra{\phi^+}{\phi^+}\right)\nonumber\\
F^o&=&F\left((\Lambda^\sA_o\otimes\Lambda^\sB)\left[(\cE\otimes \Id_\sB)[\rho]\right],(\overline\cE\otimes\Id)\left[\prjct{\phi^+}\right]\right), \nonumber
\end{eqnarray}
lead to the following bound
\begin{equation}
\arccos(\cF(\cE,\overline\cE)) \leq \arccos(F^i)+\arccos(F^o).\nonumber
\end{equation}
\end{proposition}

\begin{proof}

First note that the same map $\Lambda_\sB$ appears in both equations. Hence we get rid of all the external systems on the side $\sB$ and introduce the state $\rho' = (\Id_\sA\otimes \Lambda_\sB)[\rho] \in L(\cH_\sA \otimes \cH_i)$ where the dimension of $\cH_\sA$ is unspecified but the output subsystems on side $\sB$ is already identified. Expressing the two fidelities $F^i$ and $F^o$ with $\rho'$ gives
\begin{eqnarray}\label{eq: Fi app B}
F^i&=&F\left((\widetilde{\Lambda}_i^\sA\otimes\Id) [\rho'],\ketbra{\phi^+}{\phi^+}\right)\\
F^o&=&F\left((\Lambda_o^\sA\circ\cE\otimes\Id)[\rho'],(\overline\cE\otimes\Id)\left[\prjct{\phi^+}\right]\right). 
\label{eq:prp34}
\end{eqnarray}
Using Proposition \ref{co:lemma1} for $\cH_\t{sys} =\cH_i$, we can express the first fidelity as
\be\label{eq:prp3.1}
F^i= F\Big((\widetilde{\Phi}_i^\sA\otimes \Id)[\rho'],\prjct{\phi^+}\otimes \rho_\t{ext}^\sA \Big),
\ee
where the auxilliary state on Alice side $\rho_\t{ext}^\sA$ is by definition given by 
\be
\rho_\t{ext}^\sA= \text{Tr}_{(\cH_i^\sA\otimes \cH_i^\sB)}( \widetilde{\Phi}_i\otimes\Id[\rho'])= \text{Tr}_{\cH_i^\sA} \widetilde{\Phi}_i[ \rho_\sA]
\ee
with $\rho_\sA = \text{Tr}_\sB \,\rho$ as defined in the main text.

An isometry $\widetilde{\Phi}_i^\sA$ is a unitary embedding of the state into an Hilbert space of a larger dimension. As such, it can be decomposed as
$(\widetilde{\Phi}_i^\sA\otimes \Id) [\rho']= (\widetilde{U}_i\otimes\Id)[\rho'\otimes\prjct{\widetilde{\t{ext}}}^\sA_i]$. Plugging this into Eq.~\eqref{eq:prp3.1} gives
\begin{align}
F^i&=F\Big( (\widetilde{U}_i\otimes\Id)[\rho'\otimes\prjct{\widetilde{\t{ext}}}^\sA_i],\prjct{\phi^+}\otimes \rho_\t{ext}^\sA \Big) =\nonumber\\
&F\Big( \rho'\otimes\prjct{\widetilde{\t{ext}}}^\sA_i,(\widetilde{U}_i^{\dag}\otimes\Id)[\prjct{\phi^+}\otimes \rho_\t{ext}^\sA]\Big)\leq\nonumber\\
&F\Big(\rho', \t{Tr}_{ext} (\widetilde{U}_i^{\dag}\otimes\Id)[\prjct{\phi^+}\otimes \rho_\t{ext}^\sA]\Big)
\label{eq:prp3.2}
\end{align}
where we used the invariance of fidelity under unitary transformations and the fact that it can only increase when subsystems are traced out. Note that the right term in the last fidelity can be simply written as
\begin{align}\label{eq: Inj map}
&\text{Tr}_{ext} (\widetilde{U}_i^{\dag}\otimes\Id)[\prjct{\phi^+}\otimes \text{Tr}_{\cH_i^\sA} \widetilde{\Phi}_i[ \rho_\sA]]=\nonumber\\
&\Lambda_i^\sA\otimes \Id[\prjct{\phi^+}\otimes\rho_\sA]
\end{align}
defining the injection map $\Lambda_i^\sA$. Hence, we have
\be
F\Big(\rho', \Lambda_i^\sA\otimes \Id[\prjct{\phi^+}\otimes\rho_\sA]\Big)\geq F_i.
\ee

Now we apply the map $(\Lambda_o^\sA\!\circ\cE\otimes\Id)$ on both states\footnote{Remark that both $\cE$ and $\Lambda_o^\sA$ only act nontrivially on the "physical" system $\rho'$, the additional subsystems added by the isometries are \sout{simply simply} traced out by the composed map $\Lambda_o^\sA\!\circ\cE$.} in the last equation.  The processing inequality ensures that the fidelity can only increase by doing so, therefore
\begin{align}\label{eq:prp356}
\nonumber
& F\Big((\Lambda_o^\sA\!\circ\cE\otimes \Id)[\rho'],(\Lambda_o^\sA\!\circ\cE\circ \Lambda_i^\sA \otimes\Id)[\prjct{\phi^+}\otimes \rho_\sA] \Big) \\
&\geq F^i.
\end{align}

Next, we use the equivalence of the triangle inequality for the Unlmann fidelity
\be\label{eq:trapp}
\t{arccos}(F(\varrho_1,\varrho_3))\leq \t{arccos}(F(\varrho_1,\varrho_2)) +\t{arccos}(F(\varrho_2,\varrho_3))
\ee
for the states
\begin{align}
\varrho_1 &= (\Lambda_o^\sA\circ\cE\circ\Lambda_i^\sA\otimes \Id)[\prjct{\phi^+}\otimes \rho_\sA]\\
\varrho_2 &= (\Lambda_o^\sA\circ\cE\otimes \Id)[\rho']\\
\varrho_3 &=(\overline \cE\otimes \Id)[\prjct{\phi^+}].
\end{align}
Eq.~\eqref{eq:prp34} directly gives $F(\varrho_2,\varrho_3)=F^o$ while  Eq.~\eqref{eq:prp356} gives the bound $ F(\varrho_1,\varrho_2)\geq F_i$, which in turn implies $ \t{arccos}(F(\varrho_1,\varrho_2))\leq \t{arccos}(F_i)$. This leads to
\begin{align}\nonumber
\arccos(F(\varrho_1,\varrho_2)) \leq\arccos(F^i)+\arccos(F^o).\nonumber
\end{align}
Finally, noticing that 
\begin{align}
\nonumber
\cF(\cE,\overline \cE) = F(\varrho_1,\varrho_2)
\end{align}
implies
\be
\arccos(\cF(\cE,\overline \cE)) \leq\arccos(F^i)+\arccos(F^o).\nonumber
\ee
\end{proof}

In the proof of the previous proposition, we have explicitly constructed the injection map $\Lambda_i^\sA$ identifying input subspaces on which the channel $\cE_{i\text{-}o}$ acts. This map is constructed from the extraction map $\widetilde{\Lambda}_i^\sA$ in the certificate of the input state Eq.~(2) of the main text. We now show that a similar result holds in the multipartite case where the channel $\cE$ implements an interaction between several physical systems.

\begin{proposition}\label{proposition:B2}
Let $\rho$ be a quantum state shared between the sides $\sA$ and $\sB$, $\sA$ consisting of several parties $\{A^{(k)}\}_{k=1}^n.$ Let also $\cE$ be a channel acting jointly on all the parties on side $\sA$, $\overline \cE : L\left(\bigotimes_{k=1}^n\cH_i^{(k)}\right)\to L\left(\bigotimes_{k=1}^n\cH_o^{(k)}\right)$ a reference channel, and $\ket{\phi^+}\in \cH_i\otimes\cH_i$ a maximally entangled state.
Given the maps $\widetilde{\Lambda}_i^{A^{(k)}} [\cdot]=\textnormal{Tr}_{\widetilde{\t{ext}}_k}(\widetilde{\Phi}_i^{A^{(k)}}[\cdot])$, $\Lambda_o^{A^{(k)}} [\cdot]=\textnormal{Tr}_{\t{ext}_k}(\Phi_o^{A^{(k)}}[\cdot])$ acting locally on the respective parties $A^{(k)}$  and $\Lambda^\sB[\cdot]=\textnormal{Tr}_\textnormal{ext}(\Phi^\sB[\cdot])$ acting on $\sB$, and the product maps \begin{align}
\widetilde{\Lambda}_i^\sA&=\widetilde{\Lambda}_i^{A^{(1)}}\tot\dots\tot\widetilde{\Lambda}_i^{A^{(n)}}\nonumber\\ \nonumber \Lambda_o^\sA&=\Lambda_o^{A^{(1)}}\tot\dots\tot\Lambda_o^{A^{(n)}},
\end{align}
the following two equations 
\begin{eqnarray}
F^i&=&F\left((\widetilde{\Lambda}^\sA_i\otimes \Lambda^\sB) [\rho],\ketbra{\phi^+}{\phi^+}\right)\nonumber\\
F^o&=&F\left((\Lambda^\sA_o\otimes\Lambda^\sB)\left[(\cE\otimes \Id_\sB)[\rho]\right],(\overline\cE\otimes\Id)\left[\prjct{\phi^+}\right]\right), \nonumber
\end{eqnarray}
imply the following bound on the fidelity between the  channels $\cE$ and $\overline\cE$
\begin{equation}\label{eq:ineqFF}
\arccos(\cF(\cE,\overline\cE)) \leq \arccos(F^i)+\arccos(F^o).
\end{equation}
In addition, the injection and extraction maps $\Lambda_i$ and $\Lambda_o$ in the definition of $\cF(\cE,\overline\cE)$ (see Eq.~(4) of the main text) also have a local structure
\begin{align}\label{eq:propr5phi}
    \Lambda_i^\sA &= \Lambda_i^{(1)}\tot\dots\tot \Lambda_i^{(n)}\\
    \Lambda_o^\sA &= \Lambda_o^{(1)}\tot \dots \tot \Lambda_o^{(n)},
\end{align}
as described in Appendix A.2.
\end{proposition}
\begin{proof}
The proof is a simple modification of the proof of Proposition \ref{proposition:B1}. It is similar until Eq.~\eqref{eq:prp3.1}
\be\label{eq:propB2}
F^i= F\Big((\widetilde{\Phi}_i^\sA\otimes \Id)[\rho'],\prjct{\phi^+}\otimes \rho_\t{ext}^\sA \Big).
\ee
At this point we use the product structure of the isometry $\widetilde{\Phi}_i^\sA= \widetilde{\Phi}_i^{A^{(1)}}\otimes\dots\otimes\widetilde{\Phi}_i^{A^{(n)}}$, where each isometry acting on the party $A^{(k)}$ can be written as $\widetilde{\Phi}_i^{A^{(k)}}[\cdot] = \widetilde{U}_i^{(k)}[(\cdot)\tot \prjct{\widetilde{\t{ext}}}_i^{(k)}]$.
Hence the unitary in \eqref{eq:prp3.2} also has a product structure $\widetilde{U}_i=\widetilde{U}_i^{(1)}\otimes\dots\otimes\widetilde{U}_i^{(n)}$, which is inherited by the injection map $\Lambda_i^\sA$ defined in Eq.~\eqref{eq: Inj map}. The map $\Lambda_o^{\sA}$ have a product structure by definition. The rest simply follows the proof of proposition \ref{proposition:B1}. 
\end{proof}

To conclude, we notice that in the case where the quantum state $\rho$ is produced by independent sources it has the form
\begin{equation}
    \rho = \rho_{A^{(1)}\text{-}B^{(1)}}\otimes \dots\otimes \rho_{A^{(n)}\text{-}B^{(n)}}\otimes\rho^*,
\end{equation}
where each state $\rho_{A^{(k)}\text{-}B^{(k)}}$ is shared between the parties $A^{(k)}$ and $B^{(k)}$, and $\rho^*$ is the initial state of the device. Consequently, the marginal state of Alice is product with respect to all subsystems 
\be
\rho_\sA=\bigotimes_{k=1}^n \rho_{A^{(k)}},
\ee
with $\rho_{A^{(k)}} = \text{Tr}_{B^{(k)}} \rho_{A^{(k)}\text{-}B^{(k)}}$. This allows to conveniently absorb each state into the local injection map $\Lambda_i^{(i)}$ performed locally by the corresponding subsystem
\be
\Lambda'^{A^{(\ell)}}_i[\cdot] = \Lambda^{A^{(\ell)}}_i[\,\cdot\otimes \rho_{A^{(\ell)}}] \quad \forall \ell=1,\dots,n
\ee

Hence, in the proof of the above proposition the state of the source $\rho_\sA$ can be absorbed in this novel definition of the injection map $\Lambda'^\sA_i$, except for the degrees of freedom that describe the initial state of the measurement device itself $\rho^*$ (which may specify, for example, that the device is plugged in).
Considering the initial internal state of the device $\rho^*$ as part of the channel $\cE$, we can write the following corollary 

\begin{corollary}
Given all the condition of the proposition \ref{proposition:B2}, if in addition the sources for each input subsystem of Alice are independent, it follows that 
\begin{equation}
\arccos(\cF'(\cE,\overline\cE)) \leq \arccos(F^i)+\arccos(F^o).\nonumber
\end{equation}
where
\be
\cF'(\cE,\overline\cE) = F\Big((\Lambda_o^\sA\circ \cE \circ {\Lambda'}_i^\sA)\otimes \Id [\prjct{\phi^+}], 
\overline{\cE}\otimes\Id [\prjct{\phi^+}]\Big)
\ee
and the maps $\Lambda_o^\sA$ and ${\Lambda'}_i^\sA$ have a tensor product structure with respect to subsystems of $\sA$.
\end{corollary}
\begin{proof}
We just presented the proof.
\end{proof}

\noindent As discussed in Appendix A.2, this result would not be  possible without the assumption of independent sources.

\vspace{20 pt}

\noindent\textbf{Appendix C} \textit{Robustness of channel certification in presence of white noise--} Here we illustrate the robustness of the bounds obtained for the fidelity of the single qubit unitary channel and the two qubit CNOT gate. To this end, we consider a simple model where each element suffers from white noise. In particular, each source produces the two qubit mixed state 
\begin{equation}
\rho(\epsilon_\text{S}) = (1-\epsilon_\text{S})\prjct{\phi^+_2} +\epsilon_\t{S} \frac{1}{4}\Id,
\end{equation}
that is, with probability $\epsilon_\t{S},$ it fails to output the maximally entangled two-qubit state $\ket{\phi^+_2}$ and produces a totally depolarized two-qubit state $\frac{1}{4}\Id$ instead. 
The measurement devices output a random result with probability $\epsilon_\t{M}$, which corresponds to replacing each Pauli by $X(Z)\to (1-\epsilon_\t{M}) X(Z)$ in the Bell operator. Similarly, the channel $\cE$ fails to perform the desired operation with probability $\epsilon_\t{C}$, in which case it outputs a totally depolarized state of the corresponding dimension.

Within such a model, the calculation of expected Bell values is straightforward and gives
\begin{align}
\beta^i = & 2\sqrt{2}(1-\epsilon_\t{S})(1-\epsilon_\t{M})^2\\
\beta^o = & 2\sqrt{2}(1-\epsilon_\t{S})(1-\epsilon_\t{M})^2(1-\epsilon_\t{C})\\
\beta^i_1 = & \beta^i_2=\beta^i\\
\gamma^o =
&(1-\epsilon_\t{S})(1-\epsilon_\t{M})^2(1-\epsilon_\t{C})\times\nonumber\\&\frac{1}{5}\Big(5 - 3 \epsilon_\t{M} (2 - \epsilon_\t{M})(1 - \epsilon_\t{S}) - 
  3 \epsilon_\t{S}\Big).
\end{align}
Fig.~\ref{fig:WhiteNoise} shows  the resulting bound on the fidelity of the identity channel (left) and the CNOT gate (right) as functions of the channel noise ($\epsilon_\t{C}=\epsilon_\t{C}^\t{Id}$ and $\epsilon_\t{C}^\t{CNOT}$ respectively), assuming that $\epsilon_\t{S}=\epsilon_\t{M}=\epsilon_\t{Setup}$.\\

\begin{figure*}[t!]
\centering
\includegraphics[width=0.85\textwidth]{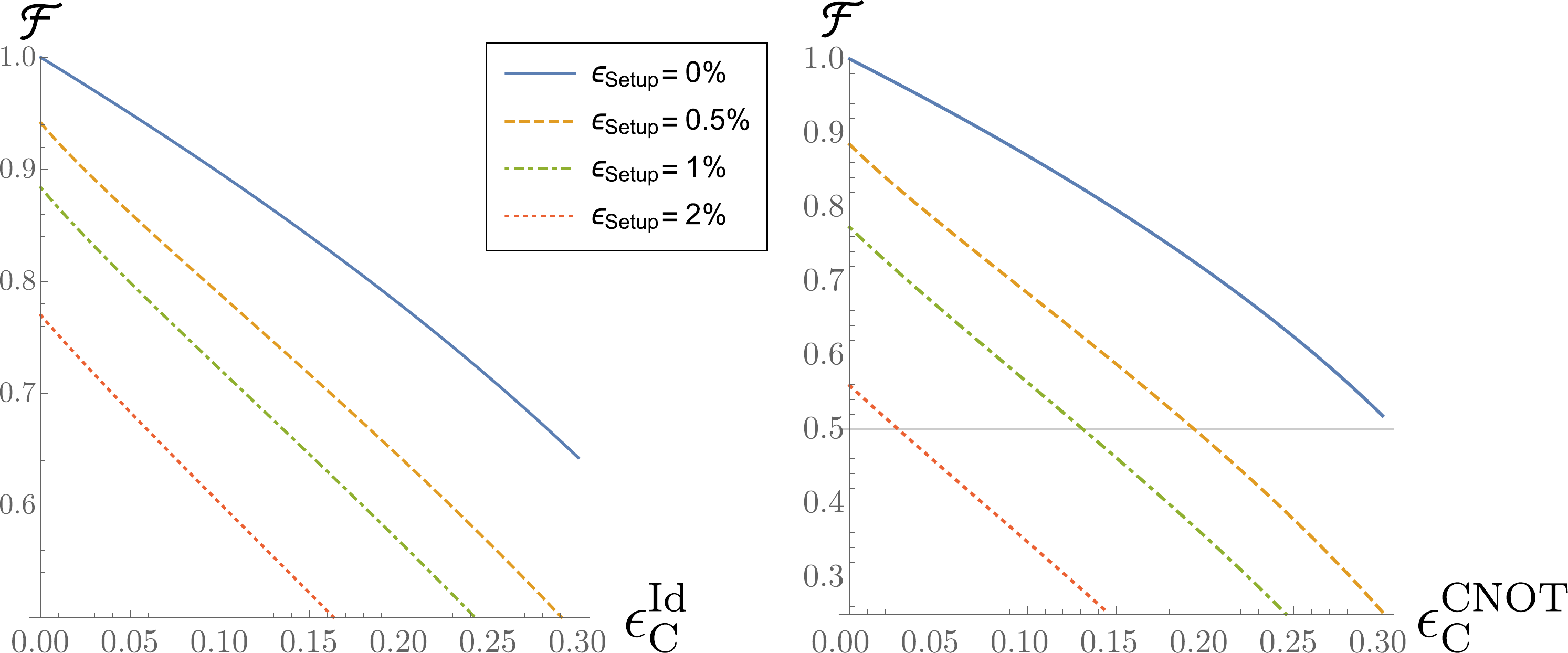}
\caption{ Device-independently certified fidelity $\cF$ for the identity channel (left) and the CNOT gate (right) in presence of white noise on the certified elements $(\epsilon_\t{C}^{\text{Id(CNOT)}})$ and on the other elements of the setup, i.e the sources and measurements $(\epsilon_\text{Setup}$).}
\label{fig:WhiteNoise}
\end{figure*}

\vspace{20 pt}

\noindent \textbf{Appendix D} \textit{Systematic approach to robustly certify N-qubit states device-independently--}  In this appendix, we describe an approach for the device-independent, robust state certification that allowed us to certify any two-qubit controlled unitary gate. We provide a systematic recipe that can be applied to arbitrary N-qubit states, even though its convergence is not guaranteed. Appendix D.2 and D.3 exploits the Jordan lemma to show how one can obtain device-independent bounds on the fidelity of N-qubit states from a given Bell test. Appendix D.4 shows how to devise Bell tests tailored to the robust certification of N-qubit states. Appendix D.5 finally applies our approach to the maximally entangled two-qubit state and the four-qubit states obtained by applying $CU_\varphi$ gates on two maximally entangled two-qubit states. Before we start, we discuss a figure of merit, the overlap, to compare two quantum states. \\

\textbf{Appendix D.1} \textit{Hilbert-Schmidt inner product and relation to the Uhlmann fidelity--} 
The Hilbert-Schmidt inner product (aka the overlap) between two states $\rho$ and $\overline \rho$ defined on the same Hilbert space is given by
\be
O(\rho,\overline \rho)=\Tr{\rho \overline \rho}.
\ee
The overlap does not satisfy some nice properties of the Uhlmann fidelity, but has the advantage to be a linear function of the quantum states and thus allows one to use efficient optimization techniques. There are several known results and methods for robust state certification~\cite{Yang14,Kaniewski16,Coladangelo17,Supic17}, that provide bounds on the overlap between the state to be certified and the target state.
As we now show, bounds on the Hilbert-Schmidt product can be directly used to bound the Uhlmann fidelity needed in Eq.~(2)-(3) of the main text.

\begin{lemma}\label{lemma:overlap}
Given two quantum states $\rho$, $\overline\rho$, the following relation holds:
\begin{equation}
F(\rho,\overline\rho) \geq \sqrt{\Tr{\rho \overline \rho}},
\end{equation}
with equality when one of the two states is pure.
\end{lemma}

\begin{proof}
With $\overline\rho=\sum_i p_i \ketbra{\overline\psi_i}{\overline\psi_i}$ we have
\begin{equation}
\begin{split}
F(\rho,\overline\rho)^2&=F(\rho,\sum_i p_i \ketbra{\overline\psi_i}{\overline\psi_i})^2\\
&\geq \sum_i p_i F(\rho,\ketbra{\overline\psi_i}{\overline\psi_i})^2\\
&= \sum_i p_i \tr{\rho\ketbra{\overline\psi_i}{\overline\psi_i}}\\
&= \tr{\rho\overline\rho}.
\end{split}
\end{equation}
Here we used the concavity of the square of the Uhlmann fidelity with respect to its second argument~\cite{Mendonca08}. The corresponding inequality is saturated when $\overline\rho$ is pure. By symmetry of the Uhlmann fidelity, equality is also achieved when $\rho$ is pure.
\end{proof}

The definition of the Hilbert-Schmidt product naturally generalizes to the case with a state $\varrho$ of the global system defined on $\cH_\t{sys}\otimes\cH_\t{ext}$ and a state $\overline \rho$ of a subsystem on $\cH_\t{sys}$ 
\be
O(\varrho ,\overline \rho) =\text{Tr}_\t{sys}\Big(\text{Tr}_\t{ext}({\varrho}) \, \overline \rho \Big),
\ee
as it is commonly done. In such a case Lemma\ref{lemma:overlap} ensures that $F(\text{Tr}_\t{ext}({\varrho}), \overline \rho)\geq O(\varrho ,\overline \rho) $ with equality when $\overline \rho$ is pure.


\textbf{Appendix D.2} \textit{Device independent certification of states using Jordan's lemma--}
In this subsection, we show how the device-independent certification of an $N$-qubit state can be reduced to an $N$-qubit problem. This reduction is only possible if certain requirements are fulfilled: (i) The Bell test has to involve at most two different measurements  per party, and (ii) these measurements must be binary. 
If these requirements are fulfilled we can use Jordan's lemma, which states the following~\cite{Scarani12}: 

\begin{lemma}
Let $X$ and $Z$ be two Hermitian operators with eigenvalues $-1$ and $+1$. Then there exists a basis in which both operators are block-diagonal,  with blocks of dimension $2\times2$ at most.
\end{lemma}

This directly implies that if a Bell operator $\cB$ consists of binary measurements then the operator $\cB$ corresponding to the Bell test can be written as
\begin{align}
    \cB = \bigoplus_{\alpha_1}\bigoplus_{\alpha_2}\hdots\bigoplus_{\alpha_N}\cB_{\alpha_1\dots \alpha_N} \, .
\end{align}
In this expression, $\cB_{\alpha_1\dots \alpha_N}$ is a $N$-qubit Bell operator, and the indices $\alpha_1,\hdots,\alpha_N$ denote the block of parties $1,\hdots,N$ respectively. Therefore, the expectation value of the Bell operator is given by
\begin{align}
    \beta = \sum_{\alpha_1,\dots,\alpha_N} \tr{\varrho_{\alpha_1\dots \alpha_N}\cB_{\alpha_1\dots \alpha_N}} \, ,
\end{align}
where $\varrho_{\alpha_1\dots \alpha_N}$ is an unnormalized $N$-qubit state corresponding to the projection of the global state  $\rho$ onto the blocks $(\alpha_1\dots \alpha_N)$. We can add the normalization by writing 
\begin{align}
    \beta &= \sum_{\alpha_1,\dots,\alpha_N} p_{\alpha_1\dots \alpha_N} \tr{\rho_{\alpha_1\dots \alpha_N}\cB_{\alpha_1\dots \alpha_N}} \nonumber\\
                     &= \sum_{\alpha_1,\dots,\alpha_N} p_{\alpha_1\dots \alpha_N} \beta_{\alpha_1\dots \alpha_N}(\rho_{\alpha_1\dots \alpha_N})
\end{align}
where $\rho_{\alpha_1\dots \alpha_N} = p_{\alpha_1\dots \alpha_N}^{-1}\varrho_{\alpha_1\dots \alpha_N}$.

This  allows to reduce the device independent certification of $N$-qubit states to a $N$-qubit problem, as we now explain.

\begin{proposition}
Consider a Bell operator of the block-diagonal form 
\begin{align}
    \cB = \bigoplus_{\alpha_1}\bigoplus_{\alpha_2}\hdots\bigoplus_{\alpha_N}\cB_{\alpha_1\dots \alpha_N}\,
\end{align}
with $\alpha_i$ labeling the block of each party $i=1\dots N$, a global state $\rho$ yielding the Bell value $\beta=\tr{\rho \cB}$, and a convex function $f$. If for each party and each subspace (block) there exist a fixed local isometry $\Phi_{\alpha_i}$, such that the overlap between any possible state $\tau$ supported on any $(\alpha_1\dots \alpha_N)$-block and the target state $\overline \rho$ is bounded by the Bell value $ \beta_{\alpha_1\dots \alpha_N}(\tau)=\tr{\tau \cB_{\alpha_1\dots \alpha_N}}$ via
\be\label{eq:Nqbound}
O\Big((\Phi_{\alpha_1}\otimes\hdots \otimes \Phi_{\alpha_N})[\tau], \overline \rho\Big)\geq f( \beta_{\alpha_1\dots \alpha_N}(\tau)),
\ee
then there is an isometry $\Phi= \Phi_1\otimes\Phi_2\otimes\hdots\Phi_N$ for which the global state satisfies
\be\label{eq:convexbound}
 O(\bm{{\Phi}}[\rho],\overline \rho)\geq f(\beta).
\ee

\label{prop:Nqubit}
\end{proposition}
\begin{proof}
For each party, let us define the isometry $\Phi_i = \bigoplus_{\alpha_i}\Phi_{\alpha_i}$, which independently maps the state in each block $\alpha_i$ onto the output (qubit) subsystem.  Under these isometries the overlap of the global state $\rho$ with the target state $\overline \rho$ satisfies 

\begin{align}
    O(&(\Phi_1\otimes\Phi_2\otimes\hdots\Phi_N)[\rho],\overline \rho) \nonumber\\
    &= O\Big((\bigoplus_{\alpha_1,\dots, \alpha_N} \Phi_{\alpha_1}\tot\hdots\!\otimes\Phi_{\alpha_N}) [\rho], \overline \rho \Big) \nonumber\\
    &= \sum_{\alpha_1,\dots, \alpha_N} p_{\alpha_1\dots \alpha_N} O\Big((\Phi_{\alpha_1}\tot\hdots\!\otimes\Phi_{\alpha_N}) [\rho_{a_1\dots a_N}],\overline\rho\Big) \nonumber\\
    &\geq \sum_{\alpha_1,\dots, \alpha_N} p_{\alpha_1\dots \alpha_N}f\Big(\beta_{\alpha_1\dots \alpha_N}(\rho_{\alpha_1\dots \alpha_N})\Big) \nonumber\\
    &\geq f\Big(\sum_{\alpha_1,\dots,\alpha_N} p_{\alpha_1\dots \alpha_N}\beta_{\alpha_1\dots \alpha_N}(\rho_{\alpha_1\dots \alpha_N})\Big) = f(\beta) \,,
\end{align}
where we used the linearity of the Hilbert-Schmidt product and the convexity of $f$.
\end{proof}
Hence, a solution of the $N$-qubit problem in the form of a convex bound as in Eq.~\eqref{eq:Nqbound}, can be directly applied to device independently certify any state $\rho$ from the observed Bell value $\beta$.
 This problem can be solved numerically, as we now show.\\


\textbf{Appendix D.3} \textit{Convex bound on the $N$-qubit fidelity--} First let us fix the form of the local observable $X_i$ and $Z_i,$ $i \in \{1,2,\hdots, N\}$. After applying Jordan's lemma and choosing a suitable coordinate system (via a local basis change for each block that is absorbed into the isometry), we can express the two observables of party $i$ in the block $\alpha_i$ as 
\begin{equation}
\begin{cases}
 X_{i,a_{\alpha_i}} = \cos(a_{\alpha_i})\sigma_x +  \sin(a_{\alpha_i})\sigma_z\\ 
 Z_{i,a_{\alpha_i}} = \cos(a_{\alpha_i})\sigma_x -  \sin(a_{\alpha_i})\sigma_z
 \end{cases}
 \label{eq:observables}
\end{equation}
where $\sigma_x$ and $\sigma_z$ are Pauli matrices, $r\in\{0,1\}$ and $a_{\alpha_i}\in[0,\frac{\pi}{2}]$. Each Bell operator $\cB_{\alpha_1\dots \alpha_N}=\cB(a_{\alpha_1},\dots,a_{\alpha_N})$ only involves terms of the form \eqref{eq:observables} and is uniquely specified by the angles $a_{\alpha_i}$. 

In order to apply the Proposition\ref{prop:Nqubit} we first fix the dependence of the local isometries on the angles $\Phi_{\alpha_i} = \Phi(a_{\alpha_i})$. Then we lower bound the overlap $O\Big((\Phi_{a_1}\otimes\hdots\otimes\Phi_{a_N})[\tau],\overline{\rho}\Big)$ between the target state $\overline \rho$ and all
$N$-qubit states $\tau$ whose Bell value $\tr{\tau \cB(a_{\alpha_1},\dots,a_{\alpha_N})}$ exceeds a certain value $\beta'$. 
This implies the following optimization:

\begin{align}
    O_\t{min}(\beta') = \min\limits_{\tau,a_1,\dots, a_N} & O\Big(\big(\Phi(a_1)\otimes\hdots\otimes\Phi(a_N)\big)[\tau],\overline{\rho}\Big)\nonumber\\
    \text{ s.t.\ \ }& \tr{\tau\cdot\cB(a_1,\dots,a_N)} \geq \beta' \nonumber\\
    &\tau \geq 0 \nonumber\\
    &\tr{\tau} = 1 \nonumber\\
    &\tau^{\dagger} = \tau \, .
    \label{eq:optimization}
\end{align}
For simplicity, we dropped the box index $\alpha_i \rightarrow i$ in these expressions.

If we fix the values of $\beta'$ and the angles $a_1,\dots,a_N$, the problem reduces to a linear optimization, which can be done very efficiently using semidefinite programming for example. So, we use the following strategy: For every $\beta'$, we fix $a_1,\dots,a_N$, we run the linear optimization and then minimize the overlap over the angles. Finally, the convex function $f(\beta)$ of Eq.~\eqref{eq:convexbound} is obtained as the convex roof of $O_\t{min}(\beta')$.\\

The optimization~\eqref{eq:optimization} is then directly formulated 
\begin{align}
O\Big(\big(\Phi(a_1)\otimes\hdots\otimes\Phi(a_N)\big)[\tau],\overline{\rho}\Big)=\nonumber\\
\t{Tr}\Big(\big(\Lambda(a_1)\otimes\hdots\otimes\Lambda(a_N)\big)[\tau] \cdot \overline{\rho}\Big),
\end{align}\\
in terms of the extraction maps $\Lambda(a_i)[\cdot]=\t{tr}_\t{ext}\Phi(a_i)[\cdot]$ induced by the isometries. The resulting $\Lambda(a_i)$  is some completely positive trace preserving (CPTP) map on a single qubit. Reciprocally any CPTP map can be dilated to an isometry, hence choosing the dependence of the local isometries $\Phi(a_i)$ on the angle amounts to choose a parametric family of extraction maps $\Lambda(a_i)$. 

To run the optimization, it remains to chose the dependence of the extraction maps $\Lambda(a_i)$ on the angles $a_i$. To do so, we follow the analytic studies presented in~\cite{Kaniewski16}, where single-qubit dephasing channels are used:
\begin{align}
    \Lambda(a)[\rho] := \frac{1+g(a)}{2}\rho + \frac{1-g(a)}{2}\Gamma_a\rho\Gamma_a \, ,
    \label{eq:Dephasing}
\end{align}
where $g(a) = (1+\sqrt{2})(\cos(a)+\sin(a)-1)$ and 
\begin{align}
    \Gamma_a = 
    \begin{cases}
        \sigma_x &\text{ for } a\in[0,\frac{\pi}{4}] \, , \\
        \sigma_z &\text{ for } a\in\,]\frac{\pi}{4},\frac{\pi}{2}] \, . 
    \end{cases} 
\end{align}

So far we have assumed that we know the Bell test  that is suitable for the robust certification of the target state $\overline \rho$. We will now give a hint on how, given a target state, the research of such a Bell test can be tackled. \\

\textbf{Appendix D.4} \textit{Devising Bell tests tailored to the robust certification of $N$-qubit states --} In this appendix, we briefly describe the approach that we used to devise Bell tests tailored to the certification of 4-qubit state family $\ket{\xi_\varphi}$. While, in principle, this approach can be used for any $N$-qubit state, its convergence is not guaranteed. In any case, it provides necessary conditions that have to be fulfilled by a Bell test to be suitable for the certification of a given state, drastically limiting the set of potential candidates. The full details will be described elsewhere~\cite{Inprep}.

Since we want the extracted fidelity to be one in the ideal case, the state $\overline \rho$ has to be the unique state  that attains the maximal quantum value of the Bell test for a  given realization of the observable. Hence, the simplest necessary condition for a Bell test to be suitable for the certification of a given state, is that this state gives the maximal Bell value for a fixed realization of the observables. As a first step, we thus consider a set of Hermitian operators that have the state $\ket{\xi_\varphi}$ for unique maximal eigenstate. Such a set of operators can be obtained by applying the gate $CU_\varphi$ on the convex sums of stabilizers of the initial Bell pairs $\ket{\phi^+_2}\ket{\phi^+_2}.$ Trivially, any such operator, except those for which the weights of some stabilizers are identically zero, has the state  $\ket{\xi_\varphi}$ as unique maximal eigenstate with eigenvalues 1. Since we want to use the Jordan lemma, we furthermore restrict ourselves to operators that can be expressed as a linear combination of correlators that only contain identity and two other Paulis per party. The Jordan lemma ensures that any such Bell test attains its maximal value for the case of qubits, i.e. when the boxes correspond to some Pauli matrices and the state is a four-qubit state. Note that different expansions of an operator in correlators (related by local rotations of the Paulis) correspond to different Bell tests. 
Similarly, a Bell test, seen as a sum of correlators, can correspond to different Bell operators, via different assignments of operators to the measurement boxes. If the aim of this test is to certify the target state, not only the state has to be the maximal eigenstate for a fixed assignment (a fixed Bell operator), but the corresponding eigenvalue, that we conveniently set to $\lambda=1$, has to be globally maximal.
Hence, as a second step, we discard the operators among all the candidates to the Bell test whose maximal eigenvalue is not a local maximum with respect to small perturbations of the measurement boxes\footnote{In our case of complementary qubit measurements $X$ and $Z$ the perturbation is  $X\to X(1-\frac{\delta^2}{2})+\delta Z$ and $Z\to Z(1-\frac{\delta^2}{2})+\delta X$.}. This can be done by standard second order eigenvalue perturbation methods: check that $\lambda$ is not perturbed at first order and all second order terms are negative. Moreover, the sensitivity of maximal eigenvalue to small misalignments of the measurement boxes (the negativity of the second order derivatives of $\lambda$) is a good indication of the fact that the Bell value is the most sensitive to a drop in state's fidelity. This is precisely what is required for a robust certificate. Therefore, as a last step we look for the Bell test for which the maximal eigenvalue is the most sensitive to perturbations in all possible directions.\\
\begin{widetext}
The resulting family of Bell tests for $\ket{\xi_\varphi}$ reads

\begin{small}
\begin{align}
\label{eq:B phi}
   \nonumber B_\varphi = \frac{1}{20}\Big( &2\sqrt{2} (E_{XX11}+E_{ZX11})+ s_\varphi (E_{X1ZX}- E_{X1XX}-E_{X1XZ}- E_{X1ZZ} +E_{Z1ZX}- E_{Z1XX}-E_{Z1XZ}- E_{Z1ZZ})
\\ \nonumber +& c_\varphi \sqrt{2}(E_{11ZX}-E_{11XZ}+E_{11XX}+E_{11ZZ})+ s_\varphi\sqrt{2} (E_{1XZX}- E_{1XXX}-E_{1XXZ}- E_{1XZZ}) 
\\ \nonumber +& c_\varphi (E_{XXZX}-E_{XXXZ}+E_{XXXX}+E_{XXZZ} +E_{ZXZX}-E_{ZXXZ}+E_{ZXXX}+E_{ZXZZ})\\
+& 2 (E_{ZZZX}-E_{ZZXZ}+E_{ZZXX}+E_{ZZZZ} -E_{XZZX}+E_{XZXZ}-E_{XZXX}-E_{XZZZ})
\Big)
\end{align}
\end{small}
with $c_\varphi = \cos(\varphi)$, $s_\varphi = \sin(\varphi)$ and correlators defined for the four parties ordering $B^{(1)}A^{(1)}A^{(2)}B^{(2)}$
\be
E_{M_1 M_2 M_3 M_4}= \sum_{b^{(1)},a^{(1)},a^{(2)},b^{(2)}} (-1)^{b^{(1)}+a^{(1)}+a^{(2)}+b^{(2)}} \t{P}\left(b^{(1)},a^{(1)},a^{(2)},b^{(2)}|M_1 M_2 M_3 M_4\right).
\ee
Here $a^{(i)}, b^{(i)}=\pm 1$ label the measurement outcomes for each party, and  $M_{k} =X$ or $Z$ labels the two possible measurement settings for party $k.$ Whenever the setting label is $M_k=1$ the corresponidng party does not measure, that is $E_{M_1 M_2 11}= \sum_{b^{(1)},a^{(1)}} (-1)^{b^{(1)}+a^{(1)}} \t{P}\left(b^{(1)},a^{(1)}|M_1 M_2\right)$. 
For the assignements
\begin{center}
\begin{tabular}{c c c  c}
$M_1$: & $X = -\frac{1}{\sqrt{2}}(\sigma_x + \sigma_z)$     & $Z = \frac{1}{\sqrt{2}}(\sigma_x - \sigma_z)$ & for party $B^{(1)}$ \\
$M_2$: & $X = -\frac{1}{\sqrt{2}}(\sigma_x + \sigma_z)$     & $Z = \frac{1}{\sqrt{2}}(\sigma_x - \sigma_z)$ & for party $A^{(1)}$  \\
$M_3$: & $X = c_\varphi \sigma_y +s_\varphi \sigma_z$ & $Z = -s_\varphi \sigma_y +c_\varphi \sigma_z$ & for party $A^{(2)}$\\
$M_4$: & $X= -\sigma_y$     & $Z = \sigma_z$ & for party $B^{(2)}$ \\
\end{tabular}
\end{center}
$B_\varphi$ attains the maximal quantum value $B_\varphi=1$ for the state $\ket{\xi_\varphi} = CU_\varphi \ket{\phi^+_2}\ket{\phi^+_2}$. Recall that the gate is applied on $A^{(1)}$ holding the control qubit and $A^{(2)}$ holding the target one. The maximally entangled states $\ket{\phi^+_2}=\frac{1}{\sqrt{2}}(\ket{00}+\ket{11})$ are shared between $A^{(1)}$ with $B^{(1)}$ and $A^{(2)}$ with $B^{(2)}$. A curious reader might be also interested to learn that the local bounds for the Bell tests $B_\varphi$ are given by
\be
B_\varphi \leq \frac{(\sqrt{2}+1)(c_\varphi + s_\varphi) + 2}{5 \sqrt{2}} \in [0.62, 0.77].
\ee
\end{widetext}

\textbf{Appendix D.5} \textit{Results--} Below we present the results of the optimization for the states discussed in the main text.

\textit{Maximally entangled two qubit state--} The maximally entangled two qubit states $\ket{\phi^+_2}=\frac{1}{\sqrt{2}}(\ket{00}+\ket{11})$ has the two following stabylizers $\sigma_x\otimes \sigma_x$ and $\sigma_z \otimes \sigma_z$ in the X-Z plane. Any operator 
\be
\cB(p) =p\, \sigma_x\otimes \sigma_x+(1-p)\sigma_z \otimes \sigma_z
\ee
with $0<p<1$ has $\ket{\phi^+_2}$ as the unique maximal eigenstate with eigenvalue $1$. Applying the approach of the previous section to the operators $\cB(p)$, we find that the Bell test
\be
\frac{1}{2\sqrt{2}}(E_{XX} + E_{XZ}+ E_{ZX}-E_{ZZ})
\ee
with $E_{M_A M_B} =\sum_{a,b=\pm 1}(-1)^{a+b}\,\t{P}\left(a,b |M_A,M_B\right)$,
is the most sensitive to the perturbation of the boxes. 
The latter corresponds to the operator $\cB(\frac{1}{2})$ for the assignments $X =\sigma_x$ and $Z=\sigma_z$ for Alice's measurement $M_A$, and  $X=\frac{1}{\sqrt{2}}(\sigma_x+\sigma_z)$ and $Z=\frac{1}{\sqrt{2}}(\sigma_x-\sigma_z)$ for Bob's measurement $M_B$.
This Bell test is proportional to the well-known CHSH inequality~\cite{CHSH69}
\be
\t{CHSH} =E_{XX} + E_{XZ}+ E_{ZX}-E_{ZZ}.
\ee
Let us now illustrate our numerical optimization method on the CHSH inequality. The results we find match the already known analytical result of Ref.~\cite{Kaniewski16}, highlighting the validity and applicability of our approach.

Eq.~\eqref{eq:observables} ensures that Alice's obsevables $X_1, Z_1$ on her $k$-block and Bob's observables $X_2, Z_2$ on his $l$-block can be written as 
\begin{align}
&\begin{cases}
    X_{1,k} &= \cos(a_k)\sigma_x + \sin(a_k)\sigma_z \\ 
    Z_{1,k} &= \cos(a_k)\sigma_x - \sin(a_k)\sigma_z
    \end{cases}\\
    &\begin{cases}
    X_{2,l} &= \cos(b_l)\sigma_x + \sin(b_l)\sigma_z \\
    Z_{2,l} &= \cos(b_l)\sigma_x - \sin(b_l)\sigma_z.\\
    \end{cases}
\end{align}
The two-qubit CHSH operator on the $(k,l)$-block is thus expressed as
\begin{align}
    \cB(a_k,b_l) = 2[&(\cos(a_k)\sigma_x + \sin(a_k)\sigma_z)\otimes\cos(b_l)\sigma_x + \nonumber\\ 
    &(\cos(a_k)\sigma_x - \sin(a_k)\sigma_z)\otimes\sin(b_l)\sigma_z] \, .
\end{align}
Up to local unitaries $\ket{\phi_2^+}$  eigenstate of $\cB(\frac\pi4,\frac\pi4)$ corresponding to the largest eigenvalue of $2\sqrt{2}$.

\begin{figure}
    \centering
    \includegraphics[width=\linewidth]{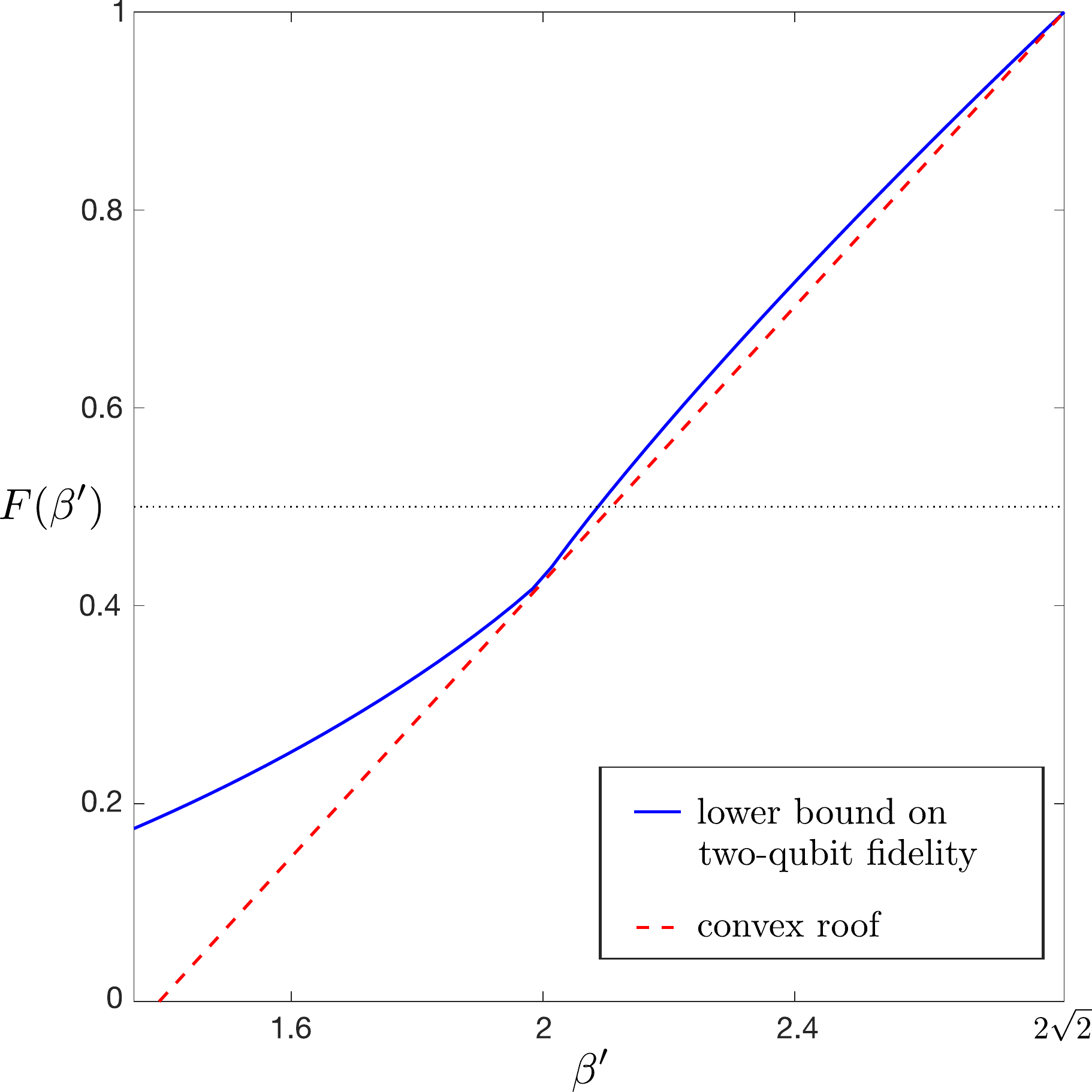}
    \caption{Application of our numerical approach to CHSH. The blue line is the result of the 2-qubit optimization. The red dashed line is the convex roof and therefore a valid device-independent lower bound on the fidelity of a two-qubit maximally entangled state. The black dashed-dotted line is the non-trivial bound of $\frac12$.}
    \label{fig:CHSH1}
\end{figure}

With the extraction maps of Eq.~\eqref{eq:Dephasing} we do the optimization~\eqref{eq:optimization} for $\beta'\in[1.35,2\sqrt{2}]$. We chose this interval to cover the whole range of fidelities $[0,1]$. As depicted in Fig.~\ref{fig:CHSH1}, the obtained bound $O_\t{min}(\beta')$ is not a convex function (solid line), but its convex roof (straight dashed line) is easy to obtain and is given by
\be
f(\overline \beta)= \frac12 \left(1+\frac{\overline \beta-\beta^*}{2\sqrt{2}-\beta^*}\right)
\ee 
with $\beta^*\approx 2.11$. Applying the Proposition \ref{prop:Nqubit} and the Lemma \ref{lemma:overlap},
we conclude that a CHSH value $\beta$ implies a lower bount on  the fidelity 
\begin{align}
    F(\t{Tr}_\t{ext}(\bm{{\Phi}}[\rho]),\prjct{\phi^+_2}) \geq \sqrt{\frac12 \left(1+\frac{\beta-\beta^*}{2\sqrt{2}-\beta^*}\right)}.
\end{align}


\textit{Four qubit entangled states $\ket{\xi_\varphi}$--} We follow the same steps for the states $\ket{\xi_\varphi}$ and corresponding Bell tests $B_\varphi$. This allows us to conclude that the Bell value $\gamma$ leads to the lower bound
\begin{align}
    F(\t{Tr}_\t{ext}(\bm{{\Phi}}[\rho]),\prjct{\xi_\varphi}) \geq \sqrt{\frac12 \left(1+\frac{\gamma-\gamma^*_\varphi}{1-\gamma^*_\varphi}\right)},
\end{align}
where the constant $\gamma^*_\varphi$ depends on the parameter $\varphi$. For the ten values of $\varphi$ \ that we analyzed we find that $\gamma_\varphi^*\in [0.8, 0.85]$, with the best case $\gamma^*_{\pi}\leq 0.795$ that corresponds to the CNOT gate, and the worst case $\gamma^*_{\frac \pi 2}\leq 0.85$ corresponding to the gate $CU_{\frac \pi 2},$ that is, the square root of a CNOT.

\vspace{20 pt}

\noindent \textbf{Appendix E} \textit{Relation to the work by Magniez et al. \cite{Magniez06}--} 
Following the questions of referees, we include here a detailed comparison between our contributions and the work by Magniez et al. \cite{Magniez06}. While the general context of both papers is the same, we believe that our results and motivations are indeed very different.

1. Our goal is to provide operational bounds for the certification of building blocks of quantum computers. Hence, our certification of a device, formalized in Eq.~\eqref{correct_def}, gives a recipe of how to use a black box in order to perform the desired quantum operation on well identified subsystems with predefined Hilbert space dimensions. In the non-ideal case
where the black box does not perform exactly as the reference operation,
we directly give a bound on the fidelity of the actual operation with respect
to the target one. The notion of equivalence in \cite{Magniez06} is defined in the other way around. It establishes that the operation performed by a black box can be reproduced by the ideal target operation if the latter is supplemented with auxiliary systems and unitary operations. In the ideal case, the two formulations can be shown to be equivalent. In the non-ideal case, however, the figure of merit proposed in~\cite{Magniez06} cannot be easily used to certify the usefulness of the black box for a particular task.

2. The strategies proposed to obtain a certificate on a device $\cE$ are different in our paper and in Ref. \cite{Magniez06}. In the latter, three measurement settings per parties are required, and one collects all the measurement statistics for the states
$\ket{\phi^+}$, $\cE \otimes\Id \ket{\phi^+}$ and $\cE \otimes \cE \ket{\phi^+}$. In particular, this requires two devices acting similarly to the target channel $\overline\cE$ in order to certify one of them. A device must then be used by each protagonist $\sA$ and $\sB$. In our case, two settings per party are required and we can certify a channel even if we only have a single instance of it: any two certificates for states $\ket{\phi^+}$ and $\cE\otimes \Id \ket{\phi^+}$ obtained with the same
measurements on side $\sB$ can be combined to certify the device $\cE$ itself.

3. The notion of “robustness” or "tolerance to noise" does not have the same meaning in our
paper and in Ref. \cite{Magniez06}. In the latter, this notion has an
information theoretic meaning. It is shown that a deviation of the statistics
by $\epsilon$ leads to a bound on the deviation of the quantum description of the black box
from the reference (tensor with something else) at the order $O(\sqrt{\epsilon})= c
\sqrt{\epsilon}$. The authors however, say nothing about the constant $c.$ This makes such a robustness statement simply unusable for any practical application. To our knowledge, the constant here is still unknown. Other self-testing results following a similar line of thought were shown to be extremely sensitive to noise. For instance, the self-testing result of~\cite{McKague12} becomes trivial for $\epsilon \simeq 4\times 10^{-5}$. The one-copy singlet state certification of~\cite{Reichardt13}, stated in lemma 4.2 of~\cite{Reichardt13bis}, is also only relevant for $\epsilon\lesssim 10^{-5}$. This estimation does not account for the overhead of the full protocol presented there which, constants aside, needs to be repeated an order of $n^{8192}$ where $n$ is the number of gates in the circuit (note that this last scaling was improved in more recent works~\cite{Hajdusek15, McKague12, Gheorghui15, Coladangelo17b}). Such robustness is currently irrelevant for any practical purpose.

With our approach, we show that a nontrivial certificate on a black box can
be obtained even with 1\% white noise on each element of the setup, see
Fig.~\ref{fig:WhiteNoise}. Recently, self-testing results with a similar noise resistance led to the first experimental self-testing of a quantum state \cite{Tan17} (a two-qubit maximally-entangled state).

\end{document}